%% file: main.tex
\pdfoutput=1
\documentclass[10pt]{article}
\usepackage{fullpage}
\usepackage[pdftex, pdfborder={0 0 0}, colorlinks=false]{hyperref}
\hypersetup{pdfborder={0 0 0}, colorlinks = true, linkcolor=black, citecolor=black, filecolor=black, urlcolor=black}
\usepackage[all]{hypcap}
\usepackage{graphicx}
\usepackage{subfigure}
\usepackage{algorithm}
\usepackage{algorithmic}
\usepackage{enumerate}
\usepackage{tikz}
\usepackage{soul}
\usepackage{etoolbox}
\usepackage{amsmath, amsthm, amssymb}
\usepackage{authblk}

\newbool{techreportbool}
\setbool{techreportbool}{false}

\newcommand{\techreport}[1]{\ifbool{techreportbool}{#1}{}}
\newcommand{\1}{{\rm 1\hspace*{-0.4ex}\rule{0.1ex}{1.52ex}\hspace*{0.2ex}}}
\newcommand{\ind}[1]{\1_{\set{#1}}}
\newcommand{\eat}[1]{}

\newcommand{\set}[1]{\{#1\}}
\newcommand{\red}[1]{\textcolor{red}{#1}}
\newcommand{\blue}[1]{\textcolor{blue}{#1}}
\newcommand{\rr}[1]{\ensuremath{\mech_{rr(#1)}}}
\newcommand{\mrr}[1]{\ensuremath{M_{rr(#1)}}}
\newcommand{\diffpriv}[1][\epsilon]{\ensuremath{{#1\text{-}}\mathfrak{diffpriv}}}
\newcommand{\mechdrop}[1][p]{\ensuremath{\mech_{\text{drop}(#1)}}}
\newcommand{\mechsamp}[1][p]{\ensuremath{\mech_{\text{sample}(#1)}}}
\newcommand{\matdrop}[1][p]{\ensuremath{M_{\text{drop}(#1)}}}
\newcommand{\matsamp}[1][p]{\ensuremath{M_{\text{sample}(#1)}}}

\DeclareMathOperator{\inp}{\mathbb{I}}  %input space
\DeclareMathOperator{\mech}{\mathfrak{M}}  %mechanism
\DeclareMathOperator{\randalg}{\mathcal{A}}  %randomized algorithm
\DeclareMathOperator{\priv}{\mathfrak{Priv}}
\DeclareMathOperator{\tdom}{\mathcal{TUP}}

\DeclareMathOperator{\range}{range}
\DeclareMathOperator{\domain}{domain}
\DeclareMathOperator{\choice}{choice}
\DeclareMathOperator{\rowcone}{rowcone}

\DeclareMathOperator{\cnf}{CNF}
\DeclareMathOperator{\data}{data}
\DeclareMathOperator{\partition}{\mathbb{P}}

\DeclareMathOperator{\parity}{parity}
\DeclareMathOperator{\hamming}{ham}
\DeclareMathOperator{\sign}{sign}
\DeclareMathOperator{\frapp}{\gamma-FRAPP}
\DeclareMathOperator{\mechsort}{\mech_{\text{sort}}}
\DeclareMathOperator{\matsort}{M_{\text{sort}}}

\newtheorem{theorem}{Theorem}[section]
\newtheorem{definition}[theorem]{Definition}
\newtheorem{lemma}[theorem]{Lemma}
\newtheorem{corollary}[theorem]{Corollary}

\newtheorem{example}[theorem]{Example}
\newtheorem{axiom}[theorem]{Axiom}

\begin{document}
\title{A Framework for Extracting Semantic Guarantees from Privacy Definitions\techreport{ (Supplementary Material)}}
%\numberofauthors{2}
%\author{
%\alignauthor Bing-Rong Lin\\
%    \affaddr{Dept. of Computer Science \& Engineering}\\
%    \affaddr{Penn State University}
%\alignauthor Daniel Kifer\\
%    \affaddr{Dept. of Computer Science \& Engineering}\\
%    \affaddr{Penn State University}
%}
\author{Bing-Rong Lin}
\author{Daniel Kifer}
\affil{Department of Computer Science \& Engineering, Pennsylvania State University, \{blin,dkifer\}@cse.psu.edu}

\maketitle
\begin{abstract}
In the field of privacy preserving data publishing, many privacy definitions have been proposed. Privacy definitions are like contracts that guide the behavior of an algorithm that takes in sensitive data and outputs non-sensitive \emph{sanitized} data. In most cases, it is not clear what these privacy definitions actually guarantee.

In this paper, we propose the first (to the best of our knowledge)
general framework for extracting semantic guarantees from privacy
definitions. These guarantees are expressed as bounds on the change in
beliefs of Bayesian attackers.

In our framework, we first restate a privacy definition in the language
of set theory and then extract from it a geometric object called the
\emph{row cone}. 
Intuitively, the row cone captures all the ways an attacker's prior
beliefs can be turned into posterior beliefs after observing an
output of an algorithm satisfying that privacy definition.
%Our framework works by restating a privacy definition in the universal
%language of set theory and then extracting from it a geometric object
%we call the \emph{row cone}. Intuitively, the row cone captures all
%the ways an attacker's prior beliefs can be turned into posterior
%beliefs after observing any output of an algorithm satisfying the
%given privacy definition. 
The row cone is a convex set and therefore has an associated set of linear inequalities. Semantic guarantees are generated by interpreting these inequalities as probabilistic statements.

Our framework can be applied to privacy definitions or to individual
algorithms to identify the types of inferences they prevent. In this
paper we use our framework to analyze the semantic privacy guarantees provided
by randomized response, FRAPP, and
several algorithms that add integer-valued noise to their inputs.
\end{abstract}
\section{Introduction}\label{sec:intro}
\input{intro}
\section{The Bird's-Eye View}\label{sec:overview}
\input{overview}

\section{Related Work}\label{sec:related}
\input{related}

\section{Consistent Normal Form and the Row Cone}\label{sec:cnf}
\input{cnf}

\section{Applications}\label{sec:applications}
   \input{applications}
       \subsection{Randomized Response}\label{sec:applications:rr}
           \input{rr} 
        \subsection{FRAPP and PRAM}\label{sec:applications:frapp}
           \input{frapp}
       % \subsection{Random Sampling}\label{sec:applications:sample}
        %   \input{sample}
       \subsection{Additive Noise} \label{sec:applications:addnoise}
      \input{noise}
        \subsection{Relaxing Privacy Definitions}\label{sec:applications:relax}
           \input{relax}

\section{Conclusions}\label{sec:conclusions}
\input{conclusions}
\bibliographystyle{abbrv}
\bibliography{shortref}
%\bibliography{ref}
\clearpage
{
\normalsize
\appendix
\input{appendix}
}
\end{document}

%% file: intro.tex
%Privacy theory of information
%what classes of information can be revealed while preventing inference about other classes of information
%what do various privacy definitions and algorithms protect
%how to fine tune

The ultimate goal of statistical privacy is to produce statistically useful sanitized data from sensitive datasets.
It has two main research thrusts: 
 developing/analyzing privacy definitions for protecting sensitive datasets, and designing algorithms that satisfy a given privacy definition while producing useful outputs. The algorithm design problem is well-posed and is the focus of most of the research activity. By contrast, privacy is a very subtle topic for which formalizing concepts is extremely challenging.

%With increased capabilities for data collection, organizations are struggling to find ways to share data with the public, researchers, and other organizations. While sharing aggregate and statistical views of the data can benefit the public, spur research, and improve business ties, sharing individual records (such as transaction data) can violate the privacy of individuals.
%
%The area of statistical privacy, which studies how algorithms filter out different types of information contained in datasets, is designed to address these issues.
%
%One goal of privacy research is to identify the kinds of views (e.g., possibly noisy aggregate or statistical data) that are safe to grant access to. In this setting, an organization first chooses a privacy definition, then finds a \emph{sanitizing} algorithm that satisfies the privacy definition. Finally, sensitive data is fed into this sanitizing algorithm and the algorithm outputs non-sensitive \emph{sanitized} data that is considered safe by the chosen privacy definition.

Analysis of privacy is important when organizations prepare to release data. When choosing a privacy definition (which subsequently guides the design of an algorithm for producing sanitized data), an organization is interested in questions such as the following. What classes of information does the privacy definition protect?
Does it offer protections that the organization is interested in? Does it offer additional protections that are not necessary (meaning that the sanitized data will contain too much distortion)? What formal protections are provided by intuitive approaches to privacy that have been collected over the past 50 years \cite{warner65:randomizedResponse}? 

In this paper we present the first (to the best of our knowledge) framework for extracting semantic guarantees from privacy definitions and individual algorithms. 
That is, instead of answering narrow questions such as ``does privacy definition Y protect X?'' the goal is to answer the more general question ``what does privacy definition Y protect?''
%Instead of trying to answer very narrow questions such as ``does a privacy definition or algorithm protect a specific type of information X?'' our goal is to address the more general question ``\emph{what} are the types of information that the privacy definition or algorithm protects?'' 
This lets an organization judge whether a privacy definition is too weak or too strong for its needs.

The framework can be used to extract guarantees about changes in the beliefs of  computationally unbounded Bayesian attackers. 
We apply our framework to several privacy definitions and algorithms for which we derived previously unknown privacy semantics -- these include randomized response \cite{warner65:randomizedResponse}, FRAPP \cite{shipraH05:frapp}/PRAM \cite{gouweleeuwKWW98:PRAM}, and several algorithms that add integer-valued noise to their inputs. It turns out that their Bayesian semantic guarantees are a consequence of their ability to protect various notions of \emph{parity} of a dataset. Since parity is frequently not a sensitive piece of information, we also show how privacy definitions can be relaxed when they are too strong.
%We derive new semantic guarantees for these privacy definitions and algorithms, showing that in the desire to protect data one often ends up protecting unnecessary information such as various notions of parity of the data. We then show how to use our framework to relax privacy definitions in an attempt to remove unnecessary protections.

Currently our framework requires a certain level of mathematical skill from the user. Tools and methodologies for reducing this burden are part of our future plans. For example, the large class of privacy definitions proposed by \cite{pufferfish} are a direct consequence of this framework; they were specifically designed to bypass the difficult parts of the framework.

The framework is based on a partial axiomatization of privacy \cite{privaxioms,privaxioms:journal}. However, the only ideas we need from \cite{privaxioms,privaxioms:journal} are two axioms and a anecdote about 2 specific  privacy definitions that do not satisfy the privacy axioms\footnote{For example, one axiom states that building a histogram from sanitized data and then releasing the histogram (instead of the sanitized data) is acceptable. This idea is widely accepted by designers of privacy definitions, yet many privacy definitions inadvertently fail to satisfy that axiom \cite{privaxioms,privaxioms:journal}.} but which imply other privacy definitions that do.

Given any privacy definition, the first step of the framework is to manipulate it using the two axioms to obtain a related privacy definition that we call the \emph{consistent normal form} (the axioms essentially remove implicit assumptions in the original privacy definition). From the consistent normal form we extract an object called the \emph{row cone} which, intuitively,  captures all the ways in which an attacker's prior belief can be turned into a posterior belief after observing an output of an algorithm that satisfies the given privacy definition. Mathematically, the row cone is represented as a convex set  and therefore has an associated collection of linear inequalities. We extract semantic guarantees by re-interpreting the coefficients of the linear inequalities as probabilities and re-interpreting the  linear inequalities themselves as statements about probabilities.

%In our framework, we first take the unifying approach of treating a privacy definition as a set of algorithms \cite{privaxioms:journal} (thus to analyze one individual algorithm, we just treat it as a privacy definition consisting of one algorithm).
%The next step is to find the \emph{consistent normal form} of this set. This is a normalization step that removes some implicit assumptions in the privacy definition. From this consistent normal form we then extract a geometric structure called the row cone. Intuitively, the row cone captures all the ways in which an attacker's prior belief can be turned into a posterior belief after observing an output from an algorithm that satisfies the given privacy definition. The row cone is actually a convex set and therefore has an associated set of linear inequalities. We extract semantic guarantees by re-interpreting the coefficients of the linear inequalities as probabilities and the linear inequalities themselves as statements about probabilities.

Our contributions are:
\begin{itemize}
%\begin{list}{\labelitemi}{\leftmargin=0.5em}
%\itemsep 0pt
%\parskip 2pt
\item A novel framework that introduces the concepts of consistent normal form and row cone and uses them to extract semantic guarantees from privacy definitions. 
\item Several applications of our framework, from which we extract previously unknown semantic guarantees for randomized response, FRAPP/PRAM, and several algorithms that add integer-valued noise to their inputs (including the Skellam distribution \cite{skellamdist} and a generalization of the geometric mechanism  \cite{universallyUtilityMaximizingPrivacyMechanisms}) .
%\end{list}
\end{itemize}

The remainder of the paper is organized as follows. We provide a detailed overview of our approach in Section \ref{sec:overview}. We discuss related work in Section \ref{sec:related}. In Section \ref{sec:cnf}, we review two privacy axioms from \cite{privaxioms,privaxioms:journal} and then we show how to use them to obtain the \emph{consistent normal form}, which removes some implicit assumptions from a privacy definition. Using the consistent normal form, we formally define the row cone (a fundamental geometric object we use for extracting semantic guarantees) in Section \ref{subsec:cnf:rowcone}. % We then proceed to discuss the relationship between our framework and the folklore concerning differential privacy \cite{diffprivacy} and syntactic privacy definitions such as $k$-anonymity \cite{samarati01:microdata} in Section \ref{sec:cnf:ex}. 
In Section \ref{sec:applications}, we then apply our framework to extract new semantic guarantees for randomized response (Section \ref{sec:applications:rr}), FRAPP/PRAM (Section \ref{sec:applications:frapp}), and noise addition algorithms (Section \ref{sec:applications:addnoise}). We discuss relaxations of privacy definitions in Section \ref{sec:applications:relax} and present conclusions in Section \ref{sec:conclusions}.
%and random sampling (Section \ref{sec:applications:sample}).

%some guarantees may be awkward or unnecessary - not fault of framework, but a consequence of privacy definitions that were created in a somewhat ad-hoc manner and without detailed thought about their full applications. This framework can help identify such problems.
%
%warmup section just for illustration, ideas are from folklore
%
%ordering on domain need only for representation as matrix
%
%shows how is related to folklore - generalizes analysis of k anonymity and differential privacy
%
%constraints serve as motivation for semantic gurantees

%% file: overview.tex
We first present some basic concepts in Section \ref{sec:overview:basic} and then provide a high-level overview of our framework in Section \ref{sec:overview:method}.
\subsection{Basic Concepts}\label{sec:overview:basic}

Let $\inp=\set{D_1, D_2,\dots}$ be the set of all possible databases. We now explain the roles played by data curators, attackers,  and privacy definitions.

\textbf{The Data Curator} owns a dataset $D\in\inp$. This dataset contains information about individuals, business secrets, etc., and therefore cannot be published as is. Thus the data curator will first choose a privacy definition and then an algorithm $\mech$ that satisfies this definition. The data curator will  apply $\mech$ to the data $D$  and will then release its output (i.e. $\mech(D)$), which we refer to as the \emph{sanitized output}. We assume that the schema of $D$ is public knowledge and that the data curator will disclose the privacy definition, release all details of the algorithm $\mech$ (except for the specific values of the random bits it used), and release the sanitized output $\mech(D)$.

\textbf{The Attacker} will use the information about the schema of $D$, the sanitized output $\mech(D)$, and knowledge of the algorithm $\mech$ to make inferences about the sensitive information contained in $D$. In our model, the attacker is computationally unbounded. The attacker may also have side information -- in the literature this is often expressed in terms of a prior distribution over possible datasets $D_i\in\inp$. In this paper we are mostly interested in guarantees against attackers who reason probabilistically and so we also assume that an attacker's side information is encapsulated in a prior distribution.

\textbf{A Privacy Definition} is often expressed as a set of algorithms that we trust (e.g., \cite{warner65:randomizedResponse,privaxioms:journal}), or a set of constraints on how an algorithm behaves (e.g., \cite{diffprivacy}), or on the type of output it produces (e.g., \cite{samarati01:microdata}). Note that treating a privacy definition as a set of algorithms is the more general approach that unifies all of these ideas \cite{privaxioms:journal} -- if a set of constraints is specified, a privacy definition  becomes the set of algorithms that satisfy  those constraints; if outputs in a certain form (such as $k$-anonymous tables \cite{samarati01:microdata})  are required, a privacy definition becomes the set of algorithms that produce those types of outputs, etc. The reason that a privacy definition should be viewed as a set of algorithms is that it allows us to manipulate privacy definitions using set theory. 

Formally, a privacy definition is 
the set of algorithms \emph{with the same input domain} that are trusted to produce nonsensitive outputs from sensitive inputs.
 We therefore use the notation $\priv$ to refer to a privacy definition and $\mech\in\priv$ to mean that the algorithm $\mech$ satisfies the privacy definition $\priv$.  

The data curator will choose a privacy definition based on what it can guarantee about the privacy of sensitive information. If a privacy definition offers too little protection (relative to the application at hand), the data curator will avoid it because sensitive information may end up being disclosed, thereby causing harm to the data curator. On the other hand, if a privacy definition offers too much protection, the resulting sanitized data may not be useful for statistical analysis. Thus it is important for the data curator to know exactly what a privacy definition guarantees.

\textbf{The Goal} is to determine what guarantees a privacy definition provides. In this paper, when we discuss semantic guarantees, we are interested in the guarantees that always hold regardless of what sanitized output is produced by an algorithm satisfying that privacy definition. We focus on computationally unbounded Bayesian attackers and look for bounds on how much their beliefs change after seeing sanitized data. It is important to note that the guarantees will depend on assumptions about the attacker's prior distribution. This is necessary, since it is well-known that without any assumptions, it is impossible to preserve privacy while providing useful sanitized data \cite{naorgame,nfl,pufferfish}.

\subsection{Overview}\label{sec:overview:method}
In a nutshell, our approach is to represent deterministic and randomized algorithms as matrices (with possibly infinitely many rows and columns) and to represent privacy definitions as sets of algorithms and hence as sets of matrices. If our goal is to analyze only a single algorithm, we simply treat it as a privacy definition (set) containing just one algorithm. The steps of our framework then require us to normalize the privacy definitions to remove some implicit assumptions (we call the result the \emph{consistent normal form}), extract the set of all rows that appear in the resulting matrices (we call this the \emph{row cone}), find linear inequalities describing those rows, reinterpret the coefficients of the linear inequalities as probabilities, and reinterpret the inequalities themselves as statements about probabilities to get semantic guarantees. In this section, we describe these steps in more detail and defer a technical exposition of the consistent normal form and row cone to Section \ref{sec:cnf}.

\subsubsection{Algorithms as matrices.}\label{sec:overview:matrix}
Since our approach relies heavily on linear algebra, it is convenient to represent algorithms as matrices. \emph{Every} algorithm $\mech$, randomized or deterministic, that runs on a digital computer can be viewed as a matrix in the following way.
An algorithm has an input domain $\inp=\set{D_1, D_2,\dots}$ consisting of datasets $D_i$, and a range $\set{\omega_1,\omega_2,\dots}$. The input domain $\inp$ and $\range(\mech)$ are necessarily countable because each $D_i\in\inp$ and $\omega_j\in\range(\mech)$ must be encoded as finite bit strings. The probability $P(\mech(D_i)=\omega_j)$ is well defined for both randomized and deterministic algorithms. The \emph{matrix representation} of an algorithm is defined as follows  (see also Figure \ref{fig:matrix}).
\begin{definition}[Matrix representation of $\mech$]\label{def:matrix} Let $\mech$ be a deterministic or randomized algorithm with domain $\inp=\set{D_1, D_2,\dots}$ and range  $\set{\omega_1,\omega_2,\dots}$. The matrix representation of $\mech$ is a (potentially infinite) matrix whose columns are indexed by $\inp$ and rows are indexed by $\range(\mech)$. The value of each entry $(i,j)$ is the quantity $P(\mech(D_j)=\omega_i)$.
\end{definition}

\begin{figure}
\begin{eqnarray*}
\bordermatrix{
 & \red{D_1} & \red{D_2} & \red{\dots} \cr
\blue{\omega_1} & P(\mech(D_1)=\omega_1) & P(\mech(D_2)=\omega_1) & \dots \cr
\blue{\omega_2} & P(\mech(D_1)=\omega_2) & P(\mech(D_2)=\omega_2) & \dots \cr
\blue{\omega_3} & P(\mech(D_1)=\omega_3) & P(\mech(D_2)=\omega_3) & \dots \cr
\vdots & \vdots &\vdots& \vdots \cr
}
\end{eqnarray*}
\label{fig:matrix}\caption{The matrix representation of $\mech$. Columns are indexed by datasets $\in\domain(\mech)$ and rows are indexed by outputs $\in\range(\mech)$.}
\end{figure}

\subsubsection{Consistent Normal Form of Privacy Definitions.}
Recall from Section \ref{sec:overview:basic} that we take the unifying view that a privacy definition is a set of algorithms (i.e., the set of algorithms that satisfy  certain constraints or produce certain types of outputs). 

Not surprisingly, there are many sets of algorithms that do not meet common expectations of what a privacy definition is \cite{privaxioms:journal}. For example, suppose that we decide to trust an algorithm $\mech$ to generate sanitized outputs from the sensitive input data $D$. Suppose we know that a researcher wants to run algorithm $\randalg$ on the sanitized data to build a histogram. If we are willing to release the sanitized output $\mech(D)$ publicly, then we should also be willing to release $\randalg(\mech(D))$. That is, if we trust $\mech$ then we should also trust $\randalg\circ\mech$ (the composition of the two algorithms). In other words, if $\mech\in \priv$, for some privacy definition $\priv$, then $\randalg\circ\mech$ should also be in $\priv$.

Many privacy definitions in the literature do not meet criteria such as this \cite{privaxioms:journal}. That is, $\mech$ may explicitly satisfy a given privacy definition but $\randalg\circ\mech$ may not. However, since the output of $\mech$ is made public and anyone can run $\randalg$ on it, these privacy definitions come with the implicit assumption that the composite algorithm $\randalg\circ\mech$ should be trusted.

Thus, given a privacy definition $\priv$, we first must expand it to include all of the algorithms we should trust (via a new application of privacy axioms).  The result of this expansion  is called the \emph{consistent normal form} and is denoted by $\cnf(\priv)$. Intuitively, $\cnf(\priv)$ is the complete set of algorithms we should trust if we accept the privacy definition $\priv$ and the privacy axioms. We describe the consistent normal form in full technical detail in Section \ref{ref:cnf:cnf}.

%Once a privacy definition is expressed as a set of algorithms, it may need further modifications to conform to general ideas about what privacy means. For example, suppose $\mech\in\priv$ is an algorithm that satisfies a privacy definition $\priv$ (i.e. $\mech$ belongs to the set of trusted algorithms) and $\randalg$ is an algorithm that can build a histogram from the output of $\mech$. This is one of the intended uses of the nonsensitive output of $\mech$ \cite{privaxioms:journal} and so privacy is not degraded. Thus we want to make sure that the composite algorithm $\randalg\circ\mech$ which first runs $\mech$ and then builds a histogram on the output of $\mech$ is also \emph{explicitly} included in our privacy definition $\priv$. In many cases this requires us to add the composite algorithm $\randalg\circ\mech$ into our set of trusted algorithms. The details of this procedure for expanding the set of trusted algorithms (without weakening the resulting privacy definition) are presented in Section \ref{sec:cnf}, and the final result is called the \emph{consistent normal form} of the privacy defintion $\priv$.

\subsubsection{The Row Cone}\label{sec:overview:rowcone}
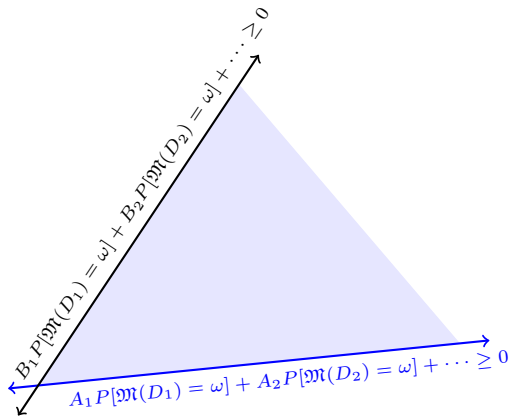
\begin{figure}[t]
\center
\begin{tikzpicture}
[scale=4, font=\scriptsize,
extreme/.style={shape=circle,draw=black!100, fill=blue!50!black!50, minimum size=2mm, inner sep=0pt, thick}]
\coordinate (zz1) at (0,0);
\coordinate (pp1) at (1.4, 1.4/10) ;
\coordinate (qq1) at (1/1.5, 1);
\draw[fill=blue!10, color=blue!10] (zz1)--(pp1)--(qq1)--(zz1);
%\coordinate [extreme,label=above right:{$(x, y)=(~P[\mech(D_1)=S],~~ P[\mech(D_2)=S]~)$}] (q) at (0.5, 0.75);
\draw[<->, thick,blue] (-0.1, -0.1/10) to  node[sloped, anchor=north west, pos=0.1]{$A_1P[\mech(D_1)=\omega] + A_2 P[\mech(D_2)=\omega]+\dots \geq 0$}  (1.5,1.5/10);
\draw[<->, thick] (-0.1/1.5, -0.1) to  node[sloped, anchor=south west, pos=0.05]{$B_1P[\mech(D_1)=\omega] + B_2 P[\mech(D_2)=\omega]+\dots \geq 0$}  (1.1/1.5,1.1);
\end{tikzpicture}
\label{fig:rowcone}
\caption{An example of a row cone (shaded) and its defining linear inequalities.}
\end{figure}
Recall that we represent algorithms as matrices (Definition \ref{def:matrix}) and privacy definitions as sets of algorithms. Therefore $\cnf(\priv)$ is really a \emph{set of matrices}. 
The row cone of $\priv$, denoted by $\rowcone(\priv)$, is the set of vectors of the form $c\vec{x}$ where $c\geq 0$ and $\vec{x}$ is a row of a matrix corresponding to some algorithm $\mech\in\cnf(\priv)$.

How does the row cone capture the semantics of $\priv$? Suppose $\mech\in\cnf(\priv)$ is one of the algorithms that we trust. Let $D$ be the true input dataset and let $\omega = \mech(D)$ be the sanitized output that we publish. A Bayesian attacker who sees output $\omega$ and is trying to derive sensitive information will need to compute the posterior distribution $P(\data=D_i~|~\mech(\data)=\omega)$ for all datasets $D_i$. This posterior distribution is a function of the attacker's prior $P(\data=D_i)$ and the vector of probabilities:
$$[P(\mech(D_1)=\omega),~~ P(\mech(D_2)=\omega),~~ \dots]$$
This vector belongs to $\rowcone(\priv)$ because it corresponds to some row of the matrix representation of $\mech$ (i.e., the row associated with output $\omega$). Note that multiplying this vector by any positive constant will leave the attacker's posterior beliefs unchanged. The row cone is essentially the set of all such probability vectors that the attacker can ever see if we use a trusted algorithm (i.e. something belonging to $\cnf(\priv)$); therefore it determines all the ways an attacker's beliefs can change (from prior to posterior). 

Thus constraints satisfied by the row cone are also constraints on how prior probabilities could be turned into posterior probabilities.
In Figure \ref{fig:rowcone} we illustrate a row cone in 2 dimensions (i.e. the input domain consists of only 2 datasets). Each vector in the row cone is represented as a point in 2-d space. 
Later in the paper, it will turn out that
%It turns out that
 the row cone is always a convex set and hence has an associated system of linear inequalities (corresponding to the intersection of halfspaces containing the row cone)  as shown in Figure \ref{fig:rowcone}.

\subsubsection{Extracting Semantic Guarantees From the Row Cone}
The row cone is a convex set (in fact, a convex cone) and so satisfies a set of linear inequalities having the forms \cite{Boyd:convex}:
\begin{eqnarray*}
A_1 P(\mech(D_1)=\omega) + A_2 P(\mech(D_2)=\omega) + \dots &\geq& 0 \text{ or}\\
A_1 P(\mech(D_1)=\omega) + A_2 P(\mech(D_2)=\omega) + \dots &=& 0\text{ or }\\
A_1 P(\mech(D_1)=\omega) + A_2 P(\mech(D_2)=\omega) + \dots &>& 0
\end{eqnarray*}
that must hold for all trusted algorithms $\mech\in\cnf(\priv)$ and sanitized outputs $\omega\in\range(\mech)$ they can produce.
The key insight is that we can re-interpret the magnitude of the coefficients $|A_1|, |A_2|,\dots$ of these linear inequalities as probabilities (dividing by $|A_1|+|A_2| + ...$ if necessary) and then re-interpret the linear inequalities as statements about prior and posterior probabilities of an attacker. We give a detailed example in Section \ref{sec:applications:rr}, where we apply our framework to randomized response. The semantic guarantees we extract then have the form: ``if the attacker's prior belongs to set $X$ then here are restrictions on the posterior probabilities the attacker can form'' (note that avoiding any assumptions on prior probabilities/knowledge is not possible if the goal is to release even marginally useful sanitized data \cite{naorgame,nfl,pufferfish}).

%% file: related.tex
\subsection{Evaluating Privacy}
Research in statistical privacy mainly focuses on developing privacy definitions and algorithms for publishing sanitized data (i.e., nonsensitive information) derived from sensitive data. To the best of our knowledge, this paper provides the first framework for extracting semantic guarantees from privacy definitions. Other work on evaluating privacy definitions looks for the presence or absence of specific vulnerabilities in privacy definitions or sanitized data. 

In the official statistics community, re-identification experiments are performed to assess whether individuals can be identified from sanitized data records \cite{nowsurvey}. In many such experiments, software is used to link sanitized data records to the original records \cite{winkler04reident}. Reiter \cite{reiterDisclosureRisk} provides a detailed example of how to apply the decision-theoretic framework of Duncan and Lambert \cite{duncanL89:disclosure} to measure disclosure risk. There are many other methods for assessing privacy for the purposes of official statistics; for surveys, see \cite{nowsurvey,willenborgW96:disclosure,willenborg00:elements}.

Other work in statistical privacy seeks to identify and exploit specific types of weaknesses that may be present in privacy definitions. 
Dwork and Naor \cite{naorgame} formally proved that it is not possible to publish anonymized data that prevents an attacker from learning information about people who are not even part of the data unless the anonymized data has very little utility or some assumptions are made about the attacker's background knowledge. 
 Lambert \cite{lambert93:disclosure} suggests that harm can occur even when an individual is linked to the wrong anonymized record (as long as the attacker's methods are plausible). Thus one of the biggest themes in privacy is preventing an attacker from linking an individual to an ``anonymized'' record \cite{dalenius86:haystack}, possibly using publicly available data \cite{sweeney02:kAnon} or other knowledge \cite{ashwin06:ldiversity}. Dinur and Nissim \cite{dinur:privacy} and later Dwork et al. \cite{dwork07:limits} showed fundamental limits to the amount of information that can be released even under very weak privacy definitions (information-theoretically and computationally \cite{Dwork09STOCOnTheComplexity}). These attacks generally work by removing noise that was added in the sanitization process \cite{karguptaDWS03:randomPerturbation,huangDC04:deriving,liu08:noiseattacks}. Ganta et al. \cite{composition08ranjit} demonstrated a composition attack where independent anonymized data releases can  be combined to breach privacy; thus a desirable property of privacy definitions is to have privacy guarantees degrade gracefully in the presence of multiple independent releases of sanitized data. The minimality attack \cite{wong:minimality} showed that privacy definitions must account for attackers who know the algorithm used to generate sanitized data; otherwise the attackers may reverse-engineer the algorithm to cause a privacy breach. The de Finetti attack \cite{kifer09attack} shows that privacy definitions based on statistical models are susceptible to attackers who make inferences using different models and use those inferences to undo the anonymization process; thus it is important to consider a wide range of inference attacks. Also, one should consider the possibility that an attacker may be able to manipulate data (e.g. by creating many new accounts in a social network) prior to its release to help break the subsequent anonymization of the data \cite{backstrom07:attackSocialNetwork}. Note also that privacy concerns can also be associated with aggregate information such as trade secrets (and not just rows in a table) \cite{cliftondefin,pufferfish}.

\subsection{Privacy Definitions}
In this section, we review some privacy definitions that will be examined in this paper.

\subsubsection{Syntactic Privacy Definitions}\label{sec:related:syntactic}
A large class of privacy definitions places restrictions on the format of the output of a randomized algorithm. Such privacy definitions are known as \emph{syntactic privacy definitions}. The prototypical syntactic privacy definition is $k$-anonymity \cite{samarati01:microdata,sweeney02:kAnon}.
In the $k$-anonymity model, a data curator first designates a set of attributes to be the \emph{quasi-identifier}. An algorithm $\mech$ then satisfies $k$-anonymity if its input is a table $T$ and its output is another table $T^*$ that is $k$-\emph{anonymous} -- for every tuple in $T^*$, there are $k-1$ other tuples that have the same value for the quasi-identifier attributes \cite{samarati01:microdata,sweeney02:kAnon}.
% For example, consider the input table in Figure \ref{fig:kana} where the attributes  \{``nationality'', ``age'', ``zip code''\} have been designated as the quasi-identifier. Given the input table from Figure \ref{fig:kana}, an algorithm $\mech$ that satisfies $3$-anonymity could output the $3$-anonymous table in Figure \ref{fig:kanb}. 
Algorithms satisfying $k$-anonymity typically work by generalizing (coarsening) attribute values. For example, if the data contains an attribute representing the age of a patient, the algorithm could generalize this attribute into age ranges of size $10$ (e.g., $[0-9], [10-19]$, etc.) or ranges of size $20$, etc. Quasi-identifier attributes are repeatedly generalized until a table $T^*$ satisfying $k$-anonymity is produced.
The rationale behind $k$-anonymity is that quasi-identifier attributes may be recorded in publicly available datasets. Linking those datasets to the original table $T$ may allow individual records to be identified, but linking to the $k$-anonymous table $T^*$ will not result in unique matches.

\subsubsection{Randomized Response}
Randomized response is a technique developed by Warner \cite{warner65:randomizedResponse} to deal with privacy issues when answering sensitive questions in a face-to-face survey. There are many variations of randomized response. One of the most popular is the following: a respondent answers truthfully with probability $p$ and lies with probability $(1-p)$, thus ensuring that the interviewer is not certain about the respondent's true answer. 
Thus the scenario where we can apply randomized response is the following:
 the input table $T$ contains 1 binary attribute and $k$ tuples. We can apply randomized response to $T$ by applying the following procedure to each tuple: flip the binary attribute with probability $1-p$. The perturbed table, which we call $T^*$, is then released. Note that randomized response is a privacy definition that consists of exactly one algorithm: the algorithm that flips each bit independently with probability $1-p$. We use our framework to extract semantic guarantees for randomized response in Section \ref{sec:applications:rr}.

\subsubsection{PRAM and FRAPP}\label{sec:related:frapp}
PRAM \cite{gouweleeuwKWW98:PRAM} and FRAPP \cite{shipraH05:frapp} are generalizations of randomized response to tables where tuples can have more than one attribute and the attributes need not be binary. PRAM can be thought of as a set of algorithms that independently perturb tuples, while FRAPP is an extension of PRAM that adds formally specified privacy restrictions to these perturbations.

Let $\tdom$ be the domain of all tuples. Each algorithm $\mech_Q$ satisfying PRAM is associated with a transition matrix $Q$ of transition probabilities, where the entry $Q_{b,a}$ is the probability $P(a\rightarrow b)$ that the algorithm changes a tuple with value $a\in\tdom$  to the value $b\in \tdom$. Given a dataset $D=\set{t_1,\dots,t_n}$, the algorithm $\mech_Q$ assigns a new value to the tuple $t_1$ according to the transition probability matrix $Q$, then it independently assigns a new value to the tuple $t_2$, etc. It is important to note that the matrix representation of $\mech_Q$ (as discussed in Section \ref{sec:overview:matrix}) \emph{is not the same} as the transition matrix $Q$. As we will discuss in Section \ref{sec:applications:frapp}, the relationship between the two is that the matrix representation of $\mech_Q$ is equal to $\bigoplus_n Q$, where $\bigoplus$ is the Kronecker product.

FRAPP, with privacy parameter $\gamma$, imposes a restriction on these algorithms. This restriction, known as $\gamma$-amplification \cite{evfimievski:limiting:breaches}, requires that the transition matrices $Q$ satisfy the constraints $\frac{Q_{b,a}}{Q_{c,a}}\leq \gamma$ for all $a,b,c\in\tdom$. This condition can also be phrased as $\frac{P(b\rightarrow a)}{P(c\rightarrow a)}\leq \gamma$. 

\subsubsection{Differential Privacy}
Differential privacy \cite{diffprivacy,dwork06Calibrating} is defined as follows:
\begin{definition} \label{def:diffpriv}A randomized algorithm $\mech$ satisfies \emph{$\epsilon$-differential privacy} if for all pairs of databases $T_1,T_2$ that differ only in the value of one tuple and for all sets $S$, $P(\mech(T_1)\in S)\leq e^\epsilon P(\mech(T_2)\in S)$.
\end{definition}

Differential privacy guarantees that the sanitized data that is output has little dependence on the value of any individual's tuple (for small values of $\epsilon$). It is known to be a weaker privacy definition than randomized response. Using our framework, we show in Section \ref{sec:applications:rrdiffp} that the difference between the two is that randomized response provides additional protection for the parity of every subset of the data.
 %Differential privacy is special in the sense that the constraints on probabilities in Definition \ref{def:diffpriv} are exactly the same as the constraints that define its row cone. We discuss this in more detail in Section \ref{sec:cnf:ex:diffp}.

%% file: cnf.tex
In this section, we formally define the \emph{consistent normal form} $\cnf(\priv)$ and $\rowcone(\priv)$ of a privacy definition $\priv$ and derive some of their important properties. These properties will later be used in Section \ref{sec:applications} to extract novel semantic guarantees for randomized response, FRAPP/PRAM, and for several algorithms (including the geometric mechanism \cite{universallyUtilityMaximizingPrivacyMechanisms}) that add integer random noise to their inputs. %As a warmup to those applications, we include Section \ref{sec:cnf:ex} to provide some simple introductory illustrations of our framework -- we compute $\cnf(\priv)$ and $\rowcone(\priv)$ for differential privacy and for syntactic anonymization methods to re-derive some semantic guarantees that are known in the folklore.

%Informally, the consistent normal form of a privacy definition is the set of algorithms we can trust if we accept that privacy definition and the row cone,  $\priv$ can be thought of as a version of $\priv$ that turns its implicit assumptions into explicit assumptions. The row cone, which is constructed from the consistent normal form, can be thought of as a summarization of the semantic guarantees provided by $\priv$ (i.e. the guarantees that hold no matter what sanitized output $\omega$ is produced by an algorithm $\mech$ satisfying $\priv$). 

\subsection{The Consistent Normal Form}\label{ref:cnf:cnf}
Recall that we treat any privacy definition $\priv$ as the set of algorithms with the same input domain. For example, we view $k$-anonymity as the set of all algorithms that produce $k$-anonymous tables \cite{samarati01:microdata}. As noted in \cite{privaxioms:journal}, such a set can often have inconsistencies.
%, such a set is often incomplete in the sense that it might not include all of the algorithms we should trust if we are prepared to accept $\priv$. 
For example, consider an algorithm $\mech$ that first transforms its input into a $k$-anonymous table and then builds a statistical model from the result and outputs the parameters of that model. Technically, this algorithm $\mech$ does not satisfy $k$-anonymity because ``model parameters'' are not a ``$k$-anonymous table.'' 
%However, it would be strange if the data curator decided to trust $k$-anonymity but not $\mech$.  if we decide to trust algorithms satisfying $k$-anonymity to protect privacy, we should also trust the algorithm $\mech$. Indeed,
However,  it would be strange if the data curator decided that releasing a $k$-anonymous table was acceptable but releasing a model built solely from that table (without any side information) was not acceptable. 
The motivation for the consistent normal form is that it makes sense to enlarge the set $\priv$ by adding $\mech$ into this set. 

It turns out that privacy axioms can help us identify the algorithms that should be added. For this purpose, we will use the following two axioms from \cite{privaxioms:journal}.

\begin{axiom}[Post-processing \cite{privaxioms:journal}]\label{ax:post}
Let $\priv$ be a privacy definition (set of algorithms). Let $\mech\in \priv$ and let $\randalg$ be any algorithm whose domain contains the range of $\mech$ and whose random bits are independent of the random bits of $\mech$. Then the composed algorithm $\randalg\circ\mech$  (which first runs $\mech$ and then runs $\randalg$ on the result) should also belong to $\priv$.\footnote{Note that if $\mech_1$ and $\mech_2$ are algorithms with the same range and domain such that $P(\mech_1(D_i)=\omega)=P(\mech_2(D_i)=\omega)$ for all $D_i\in\inp$ and $\omega\in\range(\mech_1)$, then we consider $\mech_1$ and $\mech_2$ to be equivalent.}
\end{axiom}

Note that Axiom \ref{ax:post} prevents algorithm $\randalg$ from using side information since its only input is $\mech(D)$.

\begin{axiom}[Convexity \cite{privaxioms:journal}]\label{ax:conv}
Let $\priv$ be a privacy definition (set of algorithms). Let $\mech_1\in \priv$ and $\mech_2\in \priv$ be two algorithms satisfying this privacy definition. Define the algorithm $\choice^p_{\mech_1,\mech_2}$ to be the algorithm that runs $\mech_1$ with probability $p$ and $\mech_2$ with probability $1-p$. Then $\choice^p_{\mech_1,\mech_2}$ should belong to $\priv$.
\end{axiom}

The justification in \cite{privaxioms:journal} for the convexity axiom (Axiom \ref{ax:conv}) is the following.  If both $\mech_1$ and $\mech_2$ belong to $\priv$, then both are trusted to produce sanitized data from the input data. That is, the outputs of $\mech_1$ and $\mech_2$ leave some amount of uncertainty about the input data. If the data curator randomly chooses between $\mech_1$ and $\mech_2$, the sensitive input data is protected by two layers of uncertainty: the original uncertainty added by either $\mech_1$ or $\mech_2$ and the uncertainty about which algorithm was used. Further discussion can be found in \cite{privaxioms:journal}.

Using these two axioms, we define the \emph{consistent normal form} as follows:\footnote{Note that this is a more general and useful idea than the observation in \cite{privaxioms:journal} that 2 specific variants of differential privacy do not satisfy the axioms but do imply a third variant that does satisfy the axioms.}
\begin{definition}\emph{($\cnf$).}\label{def:cnf}
Given a privacy definition $\priv$, its consistent normal form, denoted by $\cnf(\priv)$, is the smallest set of algorithms that contains $\priv$ and satisfies Axioms \ref{ax:post} and \ref{ax:conv}. 
\end{definition}

Essentially, the consistent normal form uses Axioms \ref{ax:post} and \ref{ax:conv} to turn implicit assumptions about which algorithms we trust into explicit statements -- 
%If we accept Axioms \ref{ax:post} and \ref{ax:conv}, then the consistent normal form $\cnf(\priv)$ has exactly the same privacy properties as $\priv$ in the following sense:
 if we are prepared to trust any $\mech\in\priv$ then by Axioms \ref{ax:post} and \ref{ax:conv} we should also trust any $\mech\in\cnf(\priv)$. The set $\cnf(\priv)$ is also the largest set of algorithms we should trust if we are prepared to accept $\priv$ as a privacy definition. 

The following theorem provides a useful characterization of $\cnf(\priv)$ that will help us analyze privacy definitions in Section \ref{sec:applications}.  
\begin{theorem}\label{thm:closure}
Given a privacy definition $\priv$, its consistent normal form $\cnf(\priv)$ is equivalent to the following. 
\begin{enumerate}
\item Define $\priv^{(1)}$ to be the set of all (deterministic and randomized algorithms) of the form $\randalg\circ\mech$, where $\mech\in\priv$, $\range(\mech)\subseteq\domain(\randalg)$, and the random bits of $\randalg$ and $\mech$ are independent of each other. 
\item For any positive integer $n$, finite sequence $\mech_1,\dots,\mech_n$ and probability vector $\vec{p}=(p_1,\dots,p_n)$, use the notation $\choice^{\vec p}(\mech_1,\dots,\mech_n)$ to represent the algorithm that runs $\mech_i$ with probability $p_i$. 
Define $\priv^{(2)}$ to be the set of all algorithms of the form $\choice^{\vec{p}}(\mech_1,\dots,\mech_n)$ where $n$ is a positive integer, $\mech_1,\dots,\mech_n\in\priv^{(1)}$, and $\vec{p}$ is a probability vector.
\item Set $\cnf(\priv)=\priv^{(2)}$.
\end{enumerate}
%Given a privacy definition $\priv$, its consistent normal form $\cnf(\priv)$ can be derived from the following process. 
%\begin{enumerate}
%\item Define $\priv^{(1)}$ to be the set of all (deterministic and randomized algorithms) of the form $\randalg\circ\mech$, where $\mech\in\priv$, $\range(\mech)\subseteq\domain(\randalg)$, and the random bits of $\randalg$ and $\mech$ are independent of each other. 
%\item For any positive integer $n$, finite sequence of algorithms $\mech_1,\dots,\mech_n$ and probablity vector $\vec{p}=(p_1,\dots,p_n)$, use the notation $\choice^{\vec p}(\mech_1,\dots,\mech_n)$ to represent the algorithm that runs $\mech_i$ with probability $p_i$. 
%Define $\priv^{(2)}$ to be the set of all algorithms of the form \\$\choice^p(\mech_1,\dots,\mech_n)$ where $n$ is a positive integer,\\ $\mech_1,\dots\mech_n\in\priv$, and $\vec{p}$ is a probability vector.
%\item Set $\cnf(\priv)=\priv^{(2)}$.
%\end{enumerate}
\end{theorem}
\begin{proof}
See Appendix \ref{app:close}.
\end{proof}

\begin{corollary}\label{cor:one}
If $\priv=\set{\mech}$ consists of just one algorithm, $\cnf(\priv)$ is the set of all algorithms of the form $\randalg\circ\mech$, where $\range(\mech)\subseteq\domain(\randalg)$ and the random bits in $\randalg$ and $\mech$ are independent of each other.
\end{corollary}
\begin{proof}
See Appendix \ref{app:corone}.
\end{proof}

\subsection{The Row Cone}\label{subsec:cnf:rowcone}
Having motivated the row cone in Section \ref{sec:overview:rowcone}, we now formally define it and derive its basic properties.

\begin{definition}[Row Cone]
Let $\inp=\set{D_1,D_2,\dots}$ be the set of possible input datasets and let $\priv$ be a privacy definition. The \emph{row cone} of $\priv$, denoted by $\rowcone(\priv)$, is defined as the set of vectors:
{\small
\begin{eqnarray*}
\Bigg\{\Big(c*P[\mech(D_1)=\omega],~c*P[\mech(D_2)=\omega],\dots\Big) ~:~ c\geq 0,~ \mech\in\cnf(\priv),~ \omega\in\range(\mech)\Bigg\}
\end{eqnarray*}
}
\end{definition}
Recalling the matrix representation of algorithms (as discussed in Section \ref{sec:overview:matrix} and Figure \ref{fig:matrix}), we see that a vector belongs to the row cone if and only if it is proportional to some row of the matrix representation of some trusted algorithm $\mech\in\cnf(\priv)$.

Given a $\mech\in\cnf(\priv)$ and $\omega\in\range(\mech)$, the attacker uses the vector $(P[\mech(D_1)=\omega],~P[\mech(D_2)=\omega],\dots)\in\rowcone(\priv)$ to convert the prior distribution $P(\data=D_i)$ to the posterior $P(\data=D_i~|~\mech(\data)=\omega)$. Scaling this likelihood vector by $c>0$ does not change the posterior distribution, but it does make it easier to work with the row cone.

Constraints satisfied by $\rowcone(\priv)$ are therefore constraints shared by all of the likelihood vectors $(P[\mech(D_1)=\omega],~P[\mech(D_2)=\omega],\dots)\in\rowcone(\priv)$ and therefore they constrain the ways an attacker's beliefs can change no matter what trusted algorithm $\mech\in\cnf(\priv)$ is used and what sanitized output $\omega\in\range(\mech)$ is produced.

%
%When a Bayesian attacker has a prior distribution over input datasets $D_1,D_2,\dots$ and sees a sanitized output $\omega=\mech(D)$, then the attacker's posterior beliefs are determined by the prior and by the vector of likelihoods: $(P[\mech(D_1)=\omega], P[\mech(D_2)=0],\dots)$ (which is the same as the row associated with $\omega$ in the matrix representation of $\mech$). Note that the attacker's posterior distribution is completely unchanged if this likelihood vector is rescaled by a positive constant. 
%
%Thus if the data curator uses privacy definition $\priv$ then $\rowcone(\priv)$ consists of all likelihood vectors that the attacker may possibly. Thus  $\rowcone(\priv)$ then determines all possible ways an attacker's probabilistic beliefs can change.

The row cone has an important geometric property:
\begin{theorem}\label{thm:cone}
$\rowcone(\priv)$ is a convex cone.
\end{theorem}
\begin{proof}
See Appendix \ref{app:cone}.
\end{proof}
The fact that the row cone is a convex cone means that it satisfies an associated set of linear constraints (from which we derive semantic privacy guarantees). 
For technical reasons, the treatment of these constraints differs slightly depending on whether the row cone is finite dimensional (which occurs if the number of possible datasets is finite) or infinite dimensional (if the set of possible datasets is countably infinite). We discuss this next.

\subsubsection{Finite dimensional row cones.}
A closed convex set in finite dimensions is expressible as the solution set to a system of linear inequalities \cite{Boyd:convex}. When the row cone is \emph{closed} then the linear inequalities have the form: 
\begin{eqnarray*}
A_{1,1} P[\mech(D_1)=\omega] + \dots + A_{1,n} P[\mech(D_n)=\omega] &\geq& 0\\
A_{2,1} P[\mech(D_1)=\omega] + \dots + A_{2,n} P[\mech(D_n)=\omega] &\geq& 0\\
\phantom{A_{2,1}}\vdots\phantom{(\mech(D_1)=\omega) + \dots + A_{2,n}}\vdots\phantom{ P(\mech(D_n)=\omega)} &\vdots&\vdots
\end{eqnarray*}
 (with possibly some equalities of the form $B_1 P[\mech(D_1)=\omega] + \dots + B_n P[\mech(D_n)=\omega]  = 0$ thrown in). When the row cone is not closed, it is still well-approximated by such linear inequalities: their solution set contains the row cone, and the row cone contains the solution set when the '$\geq$' in the constraints is replaced with '$>$'.
\subsubsection{Infinite dimensional row cones.}
When the domain of the data is countably infinite\footnote{We need not consider uncountably infinite domains since digital computers can only process finite bit strings, of which there are countably many.}, vectors in the row cone have infinite length since there is one component for each possible dataset. The vectors in the row cone belong to the vector space $\ell_\infty$, the set of vectors whose components are bounded. Linear constraints in this vector space can have the form: 
\begin{eqnarray}
A_1 P[\mech(D_1)=\omega] + A_2 P[\mech(D_2)=\omega] + \dots \geq 0\label{eqn:zfdc}\\
(\text{where }\sum_i |A_i|<\infty)\nonumber
\end{eqnarray}
but, if one accepts the Axiom of Choice, linear constraints are much more complicated and are generally defined via finitely additive measures \cite{analysishandbook}. On the other hand, in constructive mathematics\footnote{More precisely, mathematics based on Zermelo-Fraenkel set theory plus the Axiom of Dependent Choice \cite{analysishandbook}}, such more complicated linear constraints cannot be proven to exist (\cite{analysishandbook}, Sections 14.77, 23.10, and 27.45, and \cite{Lauwers10fa}). Therefore we  only consider the types of linear constraints shown in Equation \ref{eqn:zfdc}. 

%Thus we define:
%\begin{definition}\label{def:lincon} The \myemph{linear constraint set} of a row cone $\rowcone(\priv)$, denoted by $\lincon(\rowcone(\priv))$, is the set of all linear constraints of the form $A_1 P(\mech(D_1)=\omega) + A_2 P(\mech(D_2)=\omega) + \dots \geq 0$ (where $|A_1|+|A_2|+\dots<\infty$) that are satisfied by all elements of the row cone.
%\end{definition}
\subsubsection{Interpretation of linear constraints.}
Starting with a linear inequality of the form  $A_1 P(\mech(D_1)=\omega) +  A_2 P(\mech(D_2)=\omega) + \dots \geq 0$, we can separate out the positive coefficients, say $A_{i_1}, A_{i_2},\dots$, from the negative coefficients, say $A_{i^\prime_1}, A_{i^\prime_2},\dots$, to rewrite it in the form:
\begin{eqnarray*}
A_{i_1} P(\mech(D_{i_1})=\omega) + A_{i_2} P(\mech(D_{i_2})=\omega) + \dots \geq |A_{j^\prime_1}|~ P(\mech(D_{i^\prime_1})=\omega) + |A_{i^\prime_2} |~P(\mech(D_{i^\prime_2})=\omega) + \dots 
\end{eqnarray*}
where all of the coefficients are now positive. We can view each $A_{i_j}$ as a possible value for the prior probability $P(\data=D_{i_j})$ (or a value proportional to a prior probability). Setting $S_1=\set{D_{i_1},D_{i_2},\dots}$ and $S_2=\set{D_{i^\prime_1}, D_{i^\prime_2},\dots}$. This allows us to interpret the linear constraints as statements such as $\alpha P(\data\in S_1,\mech(\data)=\omega)\geq  P(\data\in S_2,\mech(\data)=\omega)$. Further algebraic manipulations (and a use of constants independent of $\mech$) result in statements such as:
\begin{eqnarray}
 \alpha &\geq& \frac{P(\data\in S_2~|~\mech(\data)=\omega)}{P(\data\in S_1~|~\mech(\data)=\omega)}\label{sem:1}\\
\alpha^\prime &\geq& \frac{P(\data\in S_2~|~\mech(\data)=\omega)}{P(\data\in S_1~|~\mech(\data)=\omega)} \Big/\frac{P(\data\in S_2)}{P(\data\in S_1)}\label{sem:2}
\end{eqnarray}
Equation \ref{sem:1} means that if an attacker uses a certain class of prior distributions then after seeing the sanitized data, the probability of some set $S_2$ is no more than $\alpha$ times the probability of some set $S_1$. Equation \ref{sem:2} means that if an attacker uses a certain class of priors, then the relative odds of $S_2$ vs. $S_1$ can increase by at most $\alpha^\prime$ after seeing the sanitized data\footnote{In fact, that idea has led to the creation of a  large class of privacy definitions \cite{pufferfish} as a followup to this framework; the linear constraints that characterize privacy definitions in \cite{pufferfish} are precisely the constraints of what we here call the row cone, hence all the difficult parts of the framework have been bypassed in \cite{pufferfish}.}.

Of particular importance are the sets  $S_1$ and $S_2$ of possible input datasets, whose relative probabilities are constrained by the privacy definition. In an ideal world they would correspond to something we are trying to protect (for example, $S_1$ could be the set of databases in which Bob has cancer and $S_2$ could be the set of databases in which Bob is healthy). If a privacy definition is not properly designed, $S_1$ and $S_2$ could correspond to concepts that may not need protection for certain applications (for example, $S_1$ could be the set of databases with even parity and $S_2$ could be the set of databases with odd parity).
In any case, it is important to examine existing privacy definitions and even specific algorithms to see which sets they end up protecting.

%% file: applications.tex
In this section, we present the main technical contributions of this paper -- applications of our framework for the extraction of novel semantic guarantees provided by randomized response, FRAPP/PRAM, and several algorithms (including a generalization of the geometric mechanism \cite{universallyUtilityMaximizingPrivacyMechanisms}) that add integer-valued noise to their inputs. We show randomized response and FRAPP offer particularly strong protections on different notions of parity of the input data. Since such protections are often unnecessary, we show, in Section \ref{sec:applications:relax}, how to manipulate the row cone to relax privacy definitions.

% We show that there is a close relationship between these privacy definitions and the notion of protecting the parity of every subset of tuples in the input dataset. We conclude that such privacy definitions lose utility by providing some privacy guarantees that are usually not so necessary (we believe that the same statement can be made for all algorithms that process tuples independently). We then show how privacy definitions can be relaxed by manipulating their row cones. 

We will make use of the following theorem which shows how to derive $\cnf(\priv)$ and $\rowcone(\priv)$ for a large class of privacy definitions that are based on a single algorithm.% (i.e., the privacy definition\footnote{Recall that a privacy definition is a \emph{set} of trusted algorithms.} is $\set{\mech}$ for some algorithm $\mech$).

\begin{theorem}\label{thm:invcnfinf}
Let $\inp$ be a finite or countably infinite set of possible datasets. Let $\mech^*$ be an algorithm with $\domain(\mech^*)= \inp$. Let $M^*$ be the matrix representation of $\mech^*$ (Definition \ref{def:matrix}). If $(M^*)^{-1}$ exists and the $L_1$ norm of each column of $(M^*)^{-1}$ is bounded by a constant $C$ then 
\begin{list}{\labelitemi}{\leftmargin=1em}
\itemsep 1pt
\parskip 4pt
\item[(1)] A bounded row vector $\vec{x}\in\rowcone(\set{\mech^*})$ if and only if $\vec{x}\cdot m\geq 0$ for every column $m$ of $(M^*)^{-1}$.
\item[(2)] An algorithm $\mech$, with matrix representation $M$, belongs to $\cnf(\set{\mech^*})$ if and only if the matrix $M(M^*)^{-1}$ contains no negative entries.
\item[(3)] An algorithm $\mech$, with matrix representation $M$, belongs to $\cnf(\set{\mech^*})$ if and only if every row of $M$ belongs to $\rowcone(\set{\mech^*})$.
\end{list}
\end{theorem}

%\begin{theorem}\label{thm:invcnf}
%Let $\inp=\set{D_1,\dots,D_n}$ be a finite set of possible datasets. Let $\mech^*$ be an algorithm with $\domain(\mech^*)\subseteq \inp$. Let $M^*$ be the matrix representation of $\mech^*$ (Definition \ref{def:matrix})). If $M^*$ is an invertible matrix then, denoting the $i^\text{th}$ column of $(M^*)^{-1}$ as $m^{(i)}$,
%\begin{list}{\labelitemi}{\leftmargin=1em}
%\itemsep 1pt
%\parskip 4pt
%\item A vector $\vec{x}\in\rowcone(\set{\mech^*})$ if and only if $\vec{x}\cdot m^{(i)}\geq 0$ for every column $m^{(i)}$ of $(M^*)^{-1}$.
%\item An algorithm $\mech$, with matrix representation $M$, belongs to $\cnf(\set{\mech^*})$ if and only if the matrix $M(M^*)^{-1}$ contains no negative entries.
%\item An algorithm $\mech$, with matrix representation $M$, belongs to $\cnf(\set{\mech^*})$ if and only if every row of $M$ belongs to $\rowcone(\set{\mech^*})$.
%\end{list}
%\end{theorem}
\begin{proof}
See Appendix \ref{app:invcnfinf}.
\end{proof}

Note that one of our applications,  namely the study of FRAPP/PRAM, does not satisfy the hypothesis of this theorem as it is not based on a single algorithm.
Nevertheless, this theorem still turns out to be useful for analyzing FRAPP/PRAM.

%\blue{Should we discuss uniqueness of this representation here?}
%Note that in this case the consistent normal form bears some similarities to differential privacy: for each output, there are linear constraints on the values $P(\mech(D_i)=\omega)$. In Sections \ref{sec:applications:rr}, \ref{sec:applications:frapp}, and \ref{sec:applications:sample} we show how these linear constraints can be interpreted as statements about priors and posteriors.

%% file: rr.tex
In this section, we apply our framework to extract Bayesian semantic guarantees provided by randomized response.
Recall that randomized response applies to tables with $k$ tuples and a single binary attribute. Thus each database can be represented as a bit string of length $k$. We formally define the domain of datasets and the randomized response algorithm as follows.

\begin{definition}[Domain of randomized reponse]\label{def:domrr}
Let the input domain $\inp=\set{D_1,\dots,D_{2^k}}$ be the set of all bit strings of length $k$. The bit strings are ordered in reverse lexicographic order. Thus $D_1$ is the string whose bits are all $1$ and $D_{2^k}$ is the string whose bits are all $0$.
\end{definition}

\begin{definition}[Randomized response algorithm]\label{def:rralg}
Given a privacy parameter $p\in[0,1]$, let $\rr{p}$ be the algorithm that, on input $D\in\inp$, independently flips each bit of $D$ with probability $1-p$.
\end{definition}

For example, when $k=2$ then $|\inp|=4$ and the matrix representation of $\rr{p}$ is
{\small
\begin{eqnarray*}
\bordermatrix{
 & \red{D_1=11} & \red{D_2=10} & \red{D_3=01} & \red{D_4=00} \cr
\blue{\omega_1=11} & p^2 & p(1-p) & p(1-p) & (1-p)^2 \cr
\blue{\omega_2=10} & p(1-p) & p^2 & (1-p)^2 & p(1-p) \cr
\blue{\omega_3=01} & p(1-p) & (1-p)^2 & p^2 &  p(1-p)\cr
\blue{\omega_4=00} & (1-p)^2 & p(1-p) & p(1-p) &  p^2\cr
}
\end{eqnarray*}
}

Note that randomized response, as a privacy definition, is equal to $\set{\rr{p}}$. The next lemma says that without loss of generality, we may assume that $p> 1/2$.

\begin{lemma}\label{lem:phalf}
Given a privacy parameter $p$, define $q=\max(p,1-p)$. Then
\begin{list}{\labelitemi}{\leftmargin=2em}
\itemsep 0pt
\parskip 2pt
\item $\cnf(\set{\rr{p}})=\cnf(\set{\rr{q}})$.
\item  If $p=1/2$ then $\cnf(\set{\rr{p}})$ consists of the set of algorithms whose outputs are statistically independent of their inputs (i.e. those algorithms $\mech$ where $P[\mech(D_i)=\omega]=P[\mech(D_j)=\omega]$ for all $D_i,D_j\in\inp$ and $\omega\in\range(\mech)$), and therefore attackers learn nothing from those outputs.
\end{list}
\end{lemma}
\begin{proof}
See Appendix \ref{app:phalf}.
\end{proof}

Therefore, in the remainder of this section, we assume $p>1/2$ without loss of generality.
Now we derive the consistent normal form and row cone of randomized response.
\begin{theorem}[$\cnf$ and $\rowcone$]\label{thm:rrcnf}
Given input space $\inp=\set{D_1,\dots,D_{2^k}}$ of bit strings of length $k$ and a privacy parameter $p> 1/2$,
\begin{list}{\labelitemi}{\leftmargin=1em}
\itemsep 4pt
\parskip 2pt
\item A vector $\vec{x}=(x_1,\dots,x_{2^k})\in\rowcone(\set{\rr{p}})$ if and only if for every bit string $s$ of length $k$,
$$\sum\limits_{i=1}^{2^k}p^{\hamming(s, D_i)}(p-1)^{k-\hamming(s, D_i)}x_i \geq 0$$
where $\hamming(s, D_i)$ is Hamming distance between $s$ and $D_i$.
\item An algorithm $\mech$ with matrix representation $M$ belongs to $\cnf(\set{\rr{p}})$ if and only if every row of $M$ belongs to $\rowcone(\set{\rr{p}})$.
\end{list}
\end{theorem}
\begin{proof}
See Appendix \ref{app:rrcnf}.
\end{proof}

We illustrate this theorem with our running example of tables with $k=2$ tuples.
\begin{example}\label{ex:crr2}\emph{($\cnf$ of randomized response, $k=2$).} Let $p>1/2$. With 2 tuples and one binary attribute, the domain $\inp=\set{11, 10, 01, 00}$. An algorithm $\mech$ with matrix representation $M$ belongs to the $\cnf$ of randomized response (with privacy parameter $p$) if for every vector $\vec{x}=(x_{11}, x_{10}, x_{01}, x_{00})$ that is a row of $M$, the following four constraints hold:
{\small
\begin{eqnarray}
p^2 x_{00} + (1-p)^2 x_{11} \geq p(1-p) x_{01} + p(1-p) x_{10}\label{eqn:2crr1}\\
(1-p)^2 x_{00} + p^2 x_{11} \geq p(1-p) x_{01} + p(1-p) x_{10}\label{eqn:2crr2}\\
p^2 x_{01} + (1-p)^2 x_{10} \geq p(1-p) x_{00} + p(1-p) x_{11}\label{eqn:2crr3}\\
(1-p)^2 x_{01} + p^2 x_{10} \geq p(1-p) x_{00} + p(1-p) x_{11}\label{eqn:2crr4}
\end{eqnarray}
}
\end{example}

We use Example \ref{ex:crr2} to explain the intuition behind the process of extracting Bayesian semantic guarantees from the row cone of randomized response, as given by the constraints in Equations \ref{eqn:2crr1}, \ref{eqn:2crr2}, \ref{eqn:2crr3}, and \ref{eqn:2crr4}.  Let us consider the following three attackers.

\vspace{0.5em}
\noindent\textbf{Attacker 1.}
This attacker has the prior beliefs that $P(\data=11)=p^2$, $P(\data=00)=(1-p)^2$ and $P(\data=01)=P(\data=10)=p(1-p)$, so that each bit is independent and equals $1$ with probability $p$ (this $p$ is the same as the privacy parameter $p$ in randomized response). Let us consider the effect of the constraint in Equation \ref{eqn:2crr1} on the attacker's inference. This constraint says that for all $\mech$ in the $\cnf$ of randomized response and for all $\omega\in\range(\mech)$,
{\small
\begin{eqnarray}
p^2 P[\mech(11)=\omega] + (1-p)^2P[\mech(00)=\omega] ~\geq~ p(1-p)P[\mech(01)=\omega] + p(1-p)P[\mech(10)=\omega]\label{eqn:crrex}
\end{eqnarray} 
}
Note that the coefficients in the linear constraints have the same values as the prior probabilities of the possible input datasets.
Substituting those prior beliefs into Equation \ref{eqn:crrex}, we get the constraint that for all $\omega\in\range(\mech)$:
{\small
\begin{eqnarray*}
P(\data=11) P[\mech(11)=\omega] + P(\data=00)P[\mech(00)=\omega] ~\geq~ P(\data=01)P[\mech(01)=\omega] + P(\data=10)P[\mech(10)=\omega]
\end{eqnarray*} 
}
which in turn is equal to the constraint on the attacker's belief about the joint distribution of the input and output of $\mech$:
{\small
\begin{eqnarray*}
P[\parity(\data)=0\wedge\mech(\data)=\omega] ~\geq~ P[\parity(\data)=1\wedge\mech(\data)=\omega]
\end{eqnarray*}
}
Dividing both sides by $P(\mech(\data)=\omega)$ (where $\data$ is a random variable), we get the following constraints that $\mech$ imposes on the attacker's posterior distribution:
{\small
\begin{eqnarray*}
P[\parity(\data)=0~|~\mech(\data)=\omega] ~\geq~ P[\parity(\data)=1~|~\mech(\data)=\omega]
\end{eqnarray*}
}
Thus $\mech$ guarantees that if an attacker  believes that bits in the database are generated independently with probability $p$, then after seeing the sanitized output, the attacker will believe that the true input is more likely to have even parity. Also, note that the attacker's \emph{prior} belief about even parity  (which is $p^2+(1-p)^2$) is greater than the attacker's prior belief about odd parity (which is $2p(1-p)$). Therefore $\mech$ guarantees that the attacker will not change his mind about  which parity, even or odd, is more likely.

\vspace{0.5em}
\noindent\textbf{Attacker 2.}
Now consider a different attacker who believes that the first bit in the true database is $1$ with probability $1-p$ and the second bit is $1$ with probability $p$ (both bits are still independent). Then, by similar calculations, Equation \ref{eqn:2crr3}, implies that for this attacker
{\small
\begin{eqnarray*}
P[\parity(\data)=1~|~\mech(\data)=\omega] ~\geq~ P[\parity(\data)=0~|~\mech(\data)=\omega]
\end{eqnarray*}
}
Thus, after seeing any sanitized output, the attacker will believe that the true input was more likely to have \emph{odd} parity than \emph{even} parity. This attacker's prior belief about odd parity  (which is $p^2+(1-p)^2$) is greater than this attacker's prior belief about even parity (which is $2p(1-p)$). Thus again, any $\mech$ in the $\cnf$ of randomized response will ensure that the attacker will not change his mind about the which parity is more likely.

\vspace{0.5em}
\noindent\textbf{Attacker 3.} This attacker believes that the first bit is $1$ with probability $1/2$ and believes the second bit is $1$ with probability $p$ (the bits are independent of each other). In this case, the attacker's prior beliefs are that odd parity and even parity are \emph{equally likely}. It is easy to see that now the output of $\mech$ can make the attacker change his mind about which parity is more likely (for example, consider what happens when $\rr{p}$ outputs $01$ or $00$). This is true because the attacker was so unsure about parity that even the slightest amount of evidence can change his beliefs about which parity is (slightly) more likely. However, the attacker will not change his mind about the parity of the second bit, for which he has greater confidence. This result is a consequence of Theorem \ref{thm:rrsemantics} below, which formally presents the semantic guarantees of randomized response.

\vspace{0.5em}
The difference between Attacker 3 and Attackers 1, 2 is that Attacker 3 expressed the weakest prior preference between even and odd parity (i.e. $1/2$ vs. $1/2$). Attackers 1 and 2 had stronger prior beliefs about which parity is more likely and as a result randomized response guarantees that they will not change their minds about which parity is more likely. 

The following theorem generalizes these observations to show that randomized response protects the parity of any set of bits whose prior probabilities are $\geq p$ or $\leq 1-p$ (where $p$ is the privacy parameter). It also shows that the only algorithms that have this property are the ones that belong to the trusted set $\cnf(\set{\rr{p}})$. Also note that, by Theorem \ref{thm:rrcnf},  an algorithm $\mech$ with matrix representation $M$ belongs to $\cnf(\set{\rr{p}})$ if and only if every row of $M$ belongs to $\rowcone(\set{\rr{p}})$.
% have matrix representations whose rows all belong to the row cone of randomized response (and hence, by Theorem \ref{thm:rrcnf}, belong to the $\cnf$ of randomized response). Recall that accepting Axioms \ref{ax:post} and \ref{ax:conv} means that randomized response and the $\cnf$ of randomized response provide the same privacy guarantees. 
Thus the following theorem completely characterizes the privacy guarantees provided by randomized response.\footnote{All other guarantees are a consequence of them.}

\begin{theorem}\label{thm:rrsemantics}
Let $p$ be a privacy parameter and let $\inp={D_1,\dots, D_{2^k}}$. Let $\mech$ be an algorithm that has a matrix representation whose every row  belongs to the row cone of randomized response. If the attacker believes that the bits in the data are independent and bit $i$ is equal to $1$ with probability $q_i$, then $\mech$ protects the parity of any subset of bits that have prior probability $\geq p$ or $\leq 1-p$. That is, for any subset $\set{\ell_1,\dots,\ell_m}$ of bits of the input data such that  $q_{\ell_j}\geq p~\vee~q_{\ell_j} \leq 1-p$ for $j=1,\dots, m$, the following holds:
\begin{list}{\labelitemi}{\leftmargin=0.5em}
\itemsep 4pt
\parskip 2pt
\item If $P(\parity(J)=0) \geq P(\parity(J)=1)$ then $P(\parity(J)=0~|~\mech(\data)) \geq P(\parity(J)=1~|~\mech(\data))$ 
\item If $P(\parity(J)=1) \geq P(\parity(J)=0)$ then $P(\parity(J)=1~|~\mech(\data)) \geq P(\parity(J)=0~|~\mech(\data))$ 
\end{list}
Furthermore, an algorithm $\mech$ can only provide these guarantees if every row of its matrix representation belongs to $\rowcone(\set{\rr{p}})$.
\end{theorem}
\begin{proof}
See Appendix \ref{app:rrsemantics}.
\end{proof}

In many cases, protecting the parity of an entire dataset is not  necessary  in privacy preserving applications (in fact, some people find it odd).\footnote{In this setting, we are normally interested only in the parity of individual bits since each bit corresponds to the value of one individual's record.}
Using the row cone, it is possible to relax a privacy definition to get rid of such unnecessary protections. We discuss this idea in Section \ref{sec:applications:relax}.

\subsubsection{The relationship between randomized response and differential privacy.}\label{sec:applications:rrdiffp}
When setting $\epsilon=\log\frac{p}{1-p}$ then it is well known that randomized response satisfies $\epsilon$-differential privacy. Also, for this parameter setting, differential privacy provides the same protection as randomized response for any given bit in the dataset -- a bit corresponds to the record of one individual and differential privacy would allow a bit's value to be retained with probability at most $e^\epsilon/(1+e^\epsilon)=p$ (and therefore flipped with probability $1-p$). However, Theorem \ref{thm:rrsemantics} shows that randomized response goes beyond the protection afforded by differential privacy by requiring stronger protection of the parity of larger sets of bits as well.

Note that Kasiviswanathan et al. \cite{smithlearn} proved  a learning-theoretic separation result between randomized response and differential privacy which roughly states that randomized response cannot be used to efficiently learn a problem called MASKED-PARITY. That concept of parity involves solving a linear system of equations in a $d$-dimensional vector space over the integers modulo 2. While very different from the notion of parity that we study, one direction of future work is to determine if our result about the semantic guarantees of randomized response can lead to a new proof of the result by  Kasiviswanathan et al. \cite{smithlearn}.

%% file: frapp.tex
In some cases, it may be difficult to derive the row cone of a privacy definition $\priv$. In these cases, it helps to have some notion of an approximation to a row cone from which semantic guarantees can still be extracted. One might wonder whether the Hausdorff distance \cite{hitchhiker} or some other measure of distance between sets might be a meaningful measure of the quality of an approximation. Unfortunately it is not at all clear what such a distance measure means in terms of semantic guarantees; finding a meaningful quantitative measure is an interesting open problem.  

Thus we take the following approach. If we cannot derive  $\rowcone(\priv)$, our goal becomes to find a strictly larger convex cone $r^\prime$ that contains $\rowcone(\priv)$. The reason is that any linear inequality satisfied by $r^\prime$ is also satisfied by  $\rowcone(\priv)$; the semantic interpretation of the linear inequality is then a guarantee provided by $\priv$. \textbf{Thus the approximation may lose some semantics but never generates incorrect semantics.} This idea leads to the following definition.
\begin{definition}[Approximation cone]\label{def:approxcone}
Given a privacy definition $\priv$, an \emph{approximation cone} of $\priv$ is a closed convex cone $r^\prime$ such that $\rowcone(\priv)\subseteq r^\prime$.
\end{definition}
In this section, we apply this approximation idea to FRAPP \cite{shipraH05:frapp}, which is a privacy definition based on the perturbation technique PRAM \cite{gouweleeuwKWW98:PRAM}.
Recall from Section \ref{sec:related:frapp} that the types of algorithms considered by FRAPP are algorithm $\mech_Q$ that have a transition matrix $Q$ where the $(a,b)$ entry, denoted by $P_Q(b\rightarrow a)$, is the probability that a tuple with value $b$ gets changed to $a$. The algorithm $\mech_Q$ modifies each tuple independently using this transition matrix. 

\begin{definition}[Domain of FRAPP]
Define $\tdom=\set{a_1,a_2,\dots, a_N}$ to be the domain of tuples. Choose an arbitrary ordering for these values. Define the data domain to be $\inp=\set{D_1,D_2,\dots}$ where each $D_i$ is a sequence of $k$ tuples from $\tdom$ and  the list $D_1, D_2, \dots$ is in lexicographic order.
\end{definition}

\begin{definition}[$\gamma$-FRAPP \cite{shipraH05:frapp}]
Given a privacy parameter $\gamma\geq 1$, $\frapp$ is the privacy definition containing all algorithms $\mech_Q$ that use transition matrices $Q$ with the $\gamma$-amplification property \cite{evfimievski:limiting:breaches}: for all tuple values $a,b,c\in\tdom$, $\frac{P_Q(b\rightarrow a)}{P_Q(c\rightarrow a)}\leq \gamma$.
\end{definition}

We now construct an approximation cone for $\frapp$. If $\mech_Q$ is an algorithm in $\frapp$ with transition matrix $Q$, then it is easy to see that the matrix representation of $\mech_Q$, denoted by $M_Q$, is:
$$M_Q=\bigotimes\limits_{i=1}^k Q$$
(where $k$ is the number of tuples in databases from $\inp$ and $\bigotimes$ is the Kronecker product).

Let $e_j$ be the column vector of length $N$ that has a $1$ in position $j$ and $0$ in all other positions. Write $p=\frac{\gamma}{1+\gamma}$ (so that $\gamma=\frac{p}{1-p}$). The constraints imposed on $Q$ by $\frapp$ can then be written as:
$$\forall i,j\in\set{1,\dots, N}~:~ Q (p e_i - (1-p)e_j)\succeq \vec{0}$$

where $\vec{0}$ is the vector containing only $0$ components and $\vec{a}\succeq \vec{b}$ means that $\vec{a}-\vec{b}$ has no negative components.
Therefore every vector $\vec{x}$ that is the row vector of $M_Q$, the matrix representation of $\mech_Q$, must satisfy the constraints:
\begin{eqnarray}
\forall i_1,\dots, i_k,j_1,\dots,j_k\in\set{1,\dots,N} ~:~  M_Q\left( \bigotimes\limits_{\ell=1}^k (p e_{i_\ell} - (1-p)e_{j_\ell})\right)\succeq \vec{0}\label{eqn:frapprowvector}
\end{eqnarray}

Using these constraints we can define the Kronecker approximation cone for FRAPP.
\begin{definition}\emph{(Kronecker approximation cone $\tilde{K}_p$).}
Given a privacy parameter $\gamma$, let $p=\frac{\gamma}{\gamma+1}$.
Define the \emph{Kronecker approximation cone}, denoted by $\tilde{K}_p$ to be the set of vectors $\vec{x}$ that satisfy the linear constraints in Equation \ref{eqn:frapprowvector} (where $e_{j_\ell}$ is the $j_\ell^\text{th}$ column vector of the $N\times N$ identity matrix).
%$$\vec{x}\cdot \left( \bigotimes\limits_{\ell=1}^k (p e_{i_\ell} - (1-p)e_{j_\ell})\right)\geq 0$$
%for all $i_\ell, j_\ell\in\set{1,\dots,N}$.
\end{definition}
\begin{lemma}\label{lem:frappapprox}
Let $p=\frac{\gamma}{\gamma+1}$. Then $\tilde{K}_p$ is an approximation cone for $\frapp$.
\end{lemma}
\begin{proof}
See Appendix \ref{app:frappapprox}.
\end{proof}

The connection between the approximation cone $\tilde{K}_p$ of FRAPP and $\rowcone(\rr{p})$, the row cone of randomized response, is clear once we rephrase the linear constraints that define $\rowcone(\rr{p})$ in Theorem \ref{thm:rrcnf} as follows:

\begin{eqnarray*}
\vec{x}\in \rowcone(\rr{p}) \Leftrightarrow ~ \forall ~ i_1,\dots, i_k,j_1,\dots,j_k\in\set{1,2} ~:~ \vec{x}\cdot \left( \bigotimes\limits_{\ell=1}^k (p e^\prime_{i_\ell} - (1-p)e^\prime_{j_\ell})\right)\geq 0
\end{eqnarray*}
where $e_{j_\ell}$ is the $j_\ell^\text{th}$ column vector of the $2\times 2$ identity matrix.

Thus we can use Theorem \ref{thm:rrsemantics}, which gave a semantic interpretation for randomized response to derive some of the semantic guarantees provided by FRAPP.

These guarantees are as follows. Suppose Bob is an attacker who satisfies the following conditions. 
\begin{list}{\labelitemi}{\leftmargin=0.5em}
\itemsep 4pt
\parskip 2pt
\item Bob believes that the tuples in the true dataset are independent, 
\item Bob has ruled out all but two values for the tuple of each individual. That is, for each $i$, Bob knows that the value of tuple $t_i$ is either some value $a_i\in\tdom$ or $b_i\in\tdom$.
\item For each tuple $t_i$, Bob believes that $t_i=a_i$ with probability $q_i$ and $t_i=b_i$ with probability $1-q_i$.
\end{list}
then for any subset $J$ of the tuples such that $t_i\in J$ only if $q_i\geq p=\frac{\gamma}{1+\gamma}$, then if Bob believes $P(\parity(J)=1)\geq P(\parity(J)=0)$ then after seeing output $\omega$, Bob believes $P(\parity(J)=1~|~\omega)\geq P(\parity(J)=0~|~\omega)$, and if Bob believes $P(\parity(J)=0)\geq P(\parity(J)=1)$ then $P(\parity(J)=0~|~\omega)\geq P(\parity(J)=1~|~\omega)$. Here parity can be defined arbitrarily by either treating $a_i$ or $b_i$ as a $1$ bit.

In the case of FRAPP, we also see that one of its guarantees is the protection of parity. This seems to be a general property of privacy definitions that are based on algorithms that operate on individual tuples independently.

%% file: noise.tex
In this section, we analyze a different class of algorithms -- those that add noise to their inputs. In the cases we study, the input domain is $\inp=\set{\dots, -2, -1, 0, 1, 2, \dots}$ and the algorithm being analyzed adds an integer-valued random variable to its input. In the first case that we study (Section \ref{sec:applications:negbin}), the algorithm adds a random variable of the form $Z=X-Y$ where $X$ and $Y$ have the negative binomial distribution; this includes the geometric mechanism \cite{universallyUtilityMaximizingPrivacyMechanisms} as a special case. In the second case (Section \ref{sec:applications:skellam}), the algorithm adds a random variable from a Skellam distribution \cite{skellamdist}, which has the form $Z=X-Y$ where $X$ and $Y$ have Poisson distributions.

\subsubsection{Differenced Negative Binomial Mechanism}\label{sec:applications:negbin}
The Geometric$(p)$ distribution is a probability distribution over nonnegative integers $k$ with mass function $p^k(1-p)$. The negative binomial distribution, NB$(p,r)$,  is a probability distribution over nonnegative integers $k$ with mass function ${k+r-1\choose k} p^k(1-p)^r$. It is well-known (and easy to show) that an NB$(p,r)$ random variable has the same distribution as the sum of $r$ independent Geometric$(p)$ random variables.
In order to get a distribution over the entire set of integers, we can use the difference of two independent NB$(p,r)$ random variables. This leads to the following noise addition algorithm:
\begin{definition}\label{def:negbinmech}\emph{(Differenced Negative Binomial Mechanism $\mech_{DNB(p,r)}$).}
Define  $\mech_{DNB(p,r)}$ to be the algorithm that adds $X-Y$ to its input, where $X$ and $Y$ are two independent random variables having the negative binomial distribution with parameters $p$ and $r$. We call $\mech_{DNB(p,r)}$ the \emph{differenced negative binomial mechanism}.
\end{definition}
The relationship to the geometric mechanism \cite{universallyUtilityMaximizingPrivacyMechanisms}, which adds a random integer $k$ with distribution $\frac{1-p}{1+p}p^{|k|}$, is captured in the following lemma:
\begin{lemma}\label{lem:mechgeo}
$\linebreak[0]\mech_{DNB(p,1)}$, the differenced negative binomial mechanism with $r=1$, is the geometric mechanism.
\end{lemma}
\begin{proof}
See Appendix \ref{app:mechgeo}.
\end{proof}

The following theorem gives us the row cone of the differenced negative binomial mechanism.
\begin{theorem}\label{thm:dnbrowcone}
 A  bounded row vector $\vec{x}=(\dots, \linebreak[0]x_{-2}, \linebreak[0]x_{-1}, \linebreak[0]x_0, x_1, x_2, \dots)$ belongs to $\rowcone(\set{\mech_{DNB(p,r)}})$ if for all integers $k$,
$$\forall k:~~~\sum\limits_{j=-r}^r (-1)^j f_B\left(j;\frac{p}{1+p},r\right) x_{k+j} \geq 0$$
where $p$ and $r$ are the parameters of the differenced negative binomial distribution and $f_B(\cdot;p/(1+p),r)$ is the probability mass function of the difference of two independent binomial (not negative binomial) distributions whose parameters are $p/(1+p)$ (success probability) and $r$ (number of trials).
\end{theorem}
\begin{proof}
See Appendix \ref{app:dnbrowcone}.
\end{proof}

To interpret Theorem \ref{thm:dnbrowcone} note that  (1)  the coefficients of the linear inequality are given by the distribution of the difference of two binomials, (2) the coefficients alternate in signs, and (3) for each integer $k$, the corresponding linear inequality has the coefficients shifted over by $k$ spots.

One interpretation of Theorem \ref{thm:dnbrowcone}, therefore, is that if an attacker has managed to rule out all possible inputs except $k-r, k-r+1,\dots, k+r-1, k+r$ and has a prior on these inputs that corresponds to the difference of two binomials (centered at $k$) then after seeing the sanitized output of $\mech_{DNB(p,r)}$, the attacker will believe that the set of possible inputs  $\set{\dots, k-3, k-1, k+1,\dots}$ is not more likely than $\set{\dots, k-4, k-2, k, k+2, \dots}$. Again we see a notion of protection of parity but for a smaller set of possible inputs, and note that \emph{initially} this looks like a one-sided guarantee -- the posterior probability of odd offsets from $k$ does not increase beyond the posterior probability of the even offsets from $k$.

However, what is surprising to us is that this kind of guarantee has many strong implications. To illustrate this point, consider $\mech_{DNB(p,1)}$ which is equivalent to the geometric mechanism. The linear inequalities in Theorem \ref{thm:dnbrowcone} then simplify (after some simple manipulations) to $-x_{k-1} + (p+1/p)x_k - x_{k+1}\geq 0$ which means that a mechanism must satisfy for all $k$, $-P[\mech(k-1)=\omega] +(p+1/p)P[\mech(k)=\omega]-P[\mech(k+1)=\omega]\geq 0$. Using these inequalities in the following telescoping sum, we see that they imply the familiar $\epsilon$-differential privacy constraints with $\epsilon=-\log p$ (so $e^\epsilon=1/p)$).
\begin{eqnarray*}
&& p^{-1} P[\mech(k)=\omega]-P[\mech(k-1)=\omega] \\
&=&\sum\limits_{j=0}^\infty p^j\left(-P[\mech(k-1+j)=\omega] +(p+1/p)P[\mech(k+j)=\omega]-P[\mech(k+1+j)=\omega]\right)\geq 0\\
&& p^{-1} P[\mech(k)=\omega]-P[\mech(k+1)=\omega] \\
&=&\sum\limits_{j=0}^\infty p^j\left(-P[\mech(k-1-j)=\omega] +(p+1/p)P[\mech(k-j)=\omega]-P[\mech(k+1-j)=\omega]\right)\geq 0
\end{eqnarray*}

The take-home message, we believe, from this example is that protections on parity, even one-sided protections can be very powerful (for example, we saw how the one-sided protections in Theorem \ref{thm:dnbrowcone} can imply the two-sided protections in differential privacy). Thus an interesting direction for future work is to develop methods for analyzing how different guarantees relate to each other; for example, if we protect a fact $X$, then what else do we end up protecting?

\subsubsection{Skellam Noise}\label{sec:applications:skellam}
In the previous section, we saw how (differenced) negative binomial noise was related to protections against attackers with (differenced) binomial priors, thus exhibiting a dual relationship between the binomial and negative binomial distributions. In this section, we study noise distributed according to the Skellam distribution \cite{skellamdist}, which turns out to be its own dual.

The Poisson$(\lambda)$ distribution is a probability distribution over nonnegative integers $k$ with distribution $e^{-\lambda}\frac{\lambda^k}{k!}$. A random variable $Z$ has the Skellam$(\lambda_1,\lambda_2)$ distribution if it is equal to the difference $X-Y$ of two independent random variables $X$ and $Y$ having the Poisson$(\lambda_1)$ and Poisson$(\lambda_2)$ distributions, respectively\cite{skellamdist}.

\begin{theorem}\label{thm:skellam}
Let the input domain $\inp=\set{\dots, -2, -1, 0,\linebreak[0] 1, 2, \dots}$  be the set of integers. Let $\mech_{\text{skell($\lambda_1,\lambda_2$)}}$ be the algorithm that adds to its input a random integer $k$ with the Skellam$(\lambda_1,\lambda_2)$ distribution and let $f_Z(\cdot; \lambda_1,\lambda_2)$ be the probability mass function of the Skellam$(\lambda_1,\lambda_2)$ distribution. A  bounded row vector $\vec{x}=(\dots, x_{-2}, x_{-1}, x_0, x_1, x_2, \dots)$ belongs to $\rowcone(\set{\mech_{\text{skell($\lambda_1,\lambda_2$)}}})$ if for all integers $k$,
\begin{eqnarray*}
\sum\limits_{j=-\infty}^\infty (-1)^j f_Z(j;\lambda_1,\lambda_2)x_{k+j}\geq 0
\end{eqnarray*}
\end{theorem}
\begin{proof}
See Appendix \ref{app:skellam}.
\end{proof}

As before, we see that Skellam noise protects parity if the attacker uses a Skellam prior that is shifted\footnote{i.e. the prior has the distribution of $Z+k$ where $k$ is a constant and $Z$ has the Skellam distribution.} by $k$ so that the posterior probability of the set $\set{\dots,k-3, k-1, k+1, k+3,\dots}$ cannot be higher than that of the set $\set{\dots,k-2,k,k+2,\dots}$.

\subsubsection{Other distributions.}
When the input domain is the set of integers there is a general technique for deriving the row cone corresponding to an algorithm that adds integer-valued noise to its inputs. If the noise distribution has probability mass function $f$, then the matrix representation of the noise-addition algorithm is a matrix $M$ (with rows and columns indexed by integers) whose $(i,j)$ entry is $f(i-j)$. One can take the Fourier series transform (characteristic function) $\widehat{f}(t)=\sum_{\ell=-\infty}^\infty f(\ell)e^{i\ell t}$. Let $g$ be the inverse transform of $1/\widehat{f}(t)$, if it exists. Then the inverse of the matrix $M$ is a matrix whose $(i,j)$ entries are $g(i-j)$. In combination with Theorem \ref{thm:invcnfinf}, this allows one to derive the linear constraints defining the row cone. We used this approach to derive the results of Sections \ref{sec:applications:negbin} and \ref{sec:applications:skellam} and the proof of Theorem \ref{thm:dnbrowcone} provides a formal justification for this technique.

%% file: relax.tex
 As we saw in Section \ref{sec:applications:rr}, a privacy definition $\priv$ may end up protecting more than we want.
In such cases, we can manipulate the $\rowcone(\priv)$ to relax it. This will give us a new row cone $R$
and will allow us to create a privacy definition $\priv^\prime$ of the form: $\mech\in\priv^\prime$ if and only if
every row of the matrix representation of $\mech$ belongs to $R$. 

To relax $\rowcone(\priv)$, we will replace the linear constraints that define it with weaker linear constraints. An appropriate
tool is Fourier-Motzkin elimination \cite{combopt}, which will produce a new set of linear constraints which are implied
by the old constraints. The new constraints will have fewer variables per constraint. 

We illustrate this technique by continuing Example \ref{ex:crr2} (randomized response on databases with $k=2$ tuples). Rewriting
equations \ref{eqn:2crr1} and \ref{eqn:2crr4} to isolate $x_{01}$ and setting $\alpha=p/(1-p)$, we get
\begin{eqnarray*}
%\frac{p}{1-p}x_{00}+\frac{1-p}{p} x_{11}-x_{10} &\geq& x_{01}\\
%&&\hspace{-3cm}\geq \frac{p}{1-p}x_{00} + \frac{p}{1-p}x_{11} - \frac{p^2}{(1-p)^2}x_{10}\\
%&&\hspace{-3cm}\Rightarrow x_{11}\leq \frac{p}{1-p}x_{10}
& \alpha x_{00}+ x_{11}/\alpha-x_{10} \geq x_{01} ~\geq~ \alpha x_{00} + \alpha x_{11} - \alpha^2 x_{10}\\
& \Rightarrow x_{11}\leq \alpha x_{10}
\end{eqnarray*}
Recalling that $x_{11}$ is shorthand for $P(\mech(11)=\omega)$ and $x_{10}$ is shorthand for $P(\mech(10)=\omega)$ we see that
Fourier-Motzkin elimination on the original constraints yielded one of the constraints of $(\ln\frac{p}{1-p})$-differential privacy.
Applying Fourier-Motzkin elimination on the other equations in Example \ref{ex:crr2} yields the rest of the differential privacy
constraints. Thus we see that differential privacy is a natural relaxation of randomized response.

%% file: conclusions.tex
We view privacy as a type of theory of information where the goal is to study how different algorithms filter out certain pieces of information.
To this end we proposed the first (to the best of our knowledge) framework for extracting semantic guarantees from privacy definitions.
The framework depends on the concepts of consistent normal form $\cnf(\priv)$ and $\rowcone(\priv)$. The consistent normal form corresponds to an explicit set of trusted algorithms and the row cone corresponds to the type of information that is always protected by an output of an algorithm belonging to a given privacy definition. The usefulness of these concepts comes from their geometric nature and relations to linear algebra and convex geometry.

There are many important directions for future work. These include extracting semantic guarantees that fail with a small probability, such as various probabilistic relaxations of differential privacy (e.g., \cite{ashwin08:map,chaudhuriM06:sampling}). In contrast, the row cone is only useful for finding guarantees that always hold. It is also important to study formal ways of relaxing/strengthening privacy definitions and exploring the relationships between different types of semantic guarantees.
%The framework allows us to extract the guarantees that hold no matter what sanitized output is seen by a computationally unbounded Bayesian
%attacker. The semantic guarantees are extracted by computing a geometric object called a row cone that is defined in terms of linear constraints.
%The linear constraints can be interpreted as statements about prior and posterior probabilities.
% 
%The importance of these two objects suggests that any new privacy definition $\priv$ should be presented using its normalized form $\cnf(\priv)$ or $\rowcone(\priv)$ so that its semantic guarantees can  be better evaluated.
%; that is, new privacy definitions should be worded as: ``an algorithm $\mech$ satisfies privacy definition $\priv$ if every row of the matrix representation of $\mech$ satisfies the following linear constraints ...''.

%% file: appendix.tex
%%%%%%%%%%%%%%%%%%%%%%%%%%%%%%%%%%%%%%%%%%%
%%% EAT PARITY DISCUSSION RR AND DIFFERENTIAL PRIVACY
\eat{
%%%%%%%%%%%%%%%%%%%%%%%%%%%%%%%%%%%%%%%%%%%
\section{Protection of Parity: Differential Privacy vs. Randomized Response}\label{app:parity}
\input{appparity}
%%%%%%%%%%%%%%%%%%%%%%%%%%%%%%%%%%%%%%%%%%%
%%% EAT PARITY DISCUSSION RR AND DIFFERENTIAL PRIVACY
}
%%%%%%%%%%%%%%%%%%%%%%%%%%%%%%%%%%%%%%%%%%%

\section{Proof of Theorem \lowercase{\ref{thm:closure}}}\label{app:close}
\begin{theorem}\emph{(Restatement and proof of Theorem \ref{thm:closure}).}
Given a privacy definition $\priv$, its consistent normal form $\cnf(\priv)$ is equivalent to the following.
\begin{enumerate}
\item Define $\priv^{(1)}$ to be the set of all (deterministic and randomized algorithms) of the form $\randalg\circ\mech$, where $\mech\in\priv$, $\range(\mech)\subseteq\domain(\randalg)$, and the random bits of $\randalg$ and $\mech$ are independent of each other. 
\item For any positive integer $n$, finite sequence $\mech_1,\dots,\mech_n$ and probability vector $\vec{p}=(p_1,\dots,p_n)$, use the notation $\choice^{\vec p}(\mech_1,\dots,\mech_n)$ to represent the algorithm that runs $\mech_i$ with probability $p_i$. 
Define $\priv^{(2)}$ to be the set of all algorithms of the form $\choice^{\vec{p}}(\mech_1,\dots,\mech_n)$ where $n$ is a positive integer, $\mech_1,\dots,\mech_n\in\priv^{(1)}$, and $\vec{p}$ is a probability vector.
\item Set $\cnf(\priv)=\priv^{(2)}$.
\end{enumerate}
\end{theorem}
\begin{proof}
We need to show that $\priv^{(2)}$ satisfies Axioms \ref{ax:post} and \ref{ax:conv} consistent and that any other privacy definition that satisfies both axioms and contains $\priv$ must also contain $\priv^{(2)}$. 

By construction, $\priv^{(2)}$ satisfies Axiom \ref{ax:conv} (convexity). To show that $\priv^{(2)}$ satisfies Axiom \ref{ax:post} (post-processing), choose any $\mech \in \priv^{(2)}$ and a postprocessing algorithm $\randalg$. By construction of $\priv^{(2)}$, there exists an integer $m$, a sequence of algorithms $\mech^{(1)}_1,\dots,\mech^{(1)}_m$ with each $\mech^{(1)}_i\in \priv^{(1)}$, and a probability vector $\vec{p}=(p_1,\dots,p_m)$ such that $\mech=\choice^p(\mech^{(1)}_1,\dots,\mech^{(1)}_m)$. It is easy to check that $\randalg\circ\mech=\choice^p(\randalg\circ\mech^{(1)}_1,\dots,\randalg\circ\mech^{(1)}_m)$. By construction of $\priv^{(1)}$, $\randalg\circ\mech^{(1)}_i\in\priv^{(1)}$ because $\mech^{(1)}_i\in\priv^{(1)}$. Therefore, by construction of $\priv^{(2)}$, $\randalg\circ\mech\in\priv^{(2)}$ and so $\priv^{(2)}$ satisfies Axiom \ref{ax:post} (post-processing). 

Now let $\priv^\prime$ be some privacy definition containing $\priv$ and satisfying both axioms. By Axiom \ref{ax:post} (post-processing), $\priv^{(1)}\subseteq\priv^\prime$. By Axiom \ref{ax:conv} (convexity) it follows that $\priv^{(2)}\subseteq\priv^\prime$. Therefore $\cnf(\priv)=\priv^{(2)}\subseteq \priv^\prime$.
\end{proof}

%-----------------------------------------------------------------

\section{Proof of Corollary \lowercase{\ref{cor:one}}}\label{app:corone}

\begin{corollary}\emph{(Restatement of Corollary \ref{cor:one}).}\\
If $\priv=\set{\mech}$ consists of just one algorithm, $\cnf(\priv)$ is the set of all algorithms of the form $\randalg\circ\mech$, where $\range(\mech)\subseteq\domain(\randalg)$ and the random bits in $\randalg$ and $\mech$ are independent of each other.
\end{corollary}
\begin{proof}
We use the notation defined in Theorem \ref{thm:closure}.
The corollary follows easily from process described in Theorem \ref{thm:closure} and the fact that  $$\choice^{\vec{p}}(\randalg_1\circ\mech,\dots,\randalg_n\circ\mech)=\left(\choice^{\vec{p}}(\randalg_1,\dots,\randalg_n)\right)\circ\mech$$ so that the process of computing $\cnf(\priv)$ has stopped after the first step.
\end{proof}
%-----------------------------------------------------------------

\section{Proof of Theorem \lowercase{\ref{thm:cone}}}\label{app:cone}
\begin{theorem}\emph{(Restatement and proof of Theorem \ref{thm:cone}).}
$\rowcone(\priv)$ is a convex cone.
\end{theorem}
\begin{proof}
Choose any $\vec{v}=(v_1,v_2,\dots)\in\rowcone(\priv)$. Then by definition $c\vec{v}\in\rowcone(\priv)$ for any $c\geq 0$. This takes care of the cone property so that we only need to show that $\rowcone(\priv)$ is a convex set. 

Choose any vectors $\vec{x}=(x_1,x_2,\dots)\in\rowcone(\priv)$, $\vec{y}=(y_1,y_2,\dots)\in\rowcone(\priv)$, and number $t$ such that $0\leq t\leq 1$. We show that $t\vec{x}+(1-t)\vec{y}\in\rowcone(\priv)$. If either $\vec{x}=\vec{0}$ or $\vec{y}=\vec{0}$ then we are done by the cone property we just proved. Otherwise, by definition of row cone, there exist constants $c_1,c_2>0$, algorithms $\mech_1,\mech_2\in\cnf(\priv)$, and sanitized outputs $\omega_1\in\range(\mech_1)$, $\omega_2\in \range(\mech_2)$ such that $\vec{x}/c_1$ is a row of the matrix representation of $\mech_1$ and $\vec{y}/c_2$ is a row of the matrix representation of $\mech_2$:
\begin{eqnarray*}
\vec{x}&=&\Big(c_1P[\mech_1(D_1)=\omega_1],~c_1P[\mech_1(D_2)=\omega_1],\dots\Big)\\
\vec{y}&=&\Big(c_2P[\mech_2(D_1)=\omega_2],~c_2P[\mech_2(D_2)=\omega_2],\dots\Big)
\end{eqnarray*}
Let $\randalg_1$ be the algorithm that outputs $\omega$ if its input is $\omega_1$ and $\omega^\prime$ otherwise. Similarly, let $\randalg_2$ be the algorithm that outputs $\omega$ if its input is $\omega_2$ and $\omega^\prime$ otherwise. Define $\mech_1^\prime\equiv\randalg_1\circ\mech_1$ and $\mech_2^\prime\equiv\randalg_2\circ\mech_2$. Then by Theorem \ref{thm:closure} (and the post-processing Axiom \ref{ax:post}), $\mech_1^\prime,\mech_2^\prime\in\cnf(\priv)$ and 
\begin{eqnarray*}
\vec{x}&=&\Big(c_1P[\mech^\prime_1(D_1)=\omega],~c_1P[\mech^\prime_1(D_2)=\omega],\dots\Big)\\
\vec{y}&=&\Big(c_2P[\mech^\prime_2(D_1)=\omega],~c_2P[\mech^\prime_2(D_2)=\omega],\dots\Big)
\end{eqnarray*}
Now consider the algorithm $\mech^*$ which runs $\mech_1^\prime$ with probability $\frac{tc_1}{tc_1+(1-t)c_2}$ and runs $\mech_2^\prime$ with probability $\frac{(1-t)c_2}{tc_1 + (1-t)c_2}$. By Theorem \ref{thm:closure}, $\mech^*\in\cnf(\priv)$. Then for all $i=1,2,\dots$,
\begin{eqnarray*}
P(\mech^*(D_i)=\omega) &=&\frac{tc_1P(\mech_1^\prime(D_i)=\omega)+(1-t)c_2P(\mech_2^\prime(D_i)=\omega)}{tc_1+(1-t)c_2}\\
&=& \frac{tx_i + (1-t)y_i}{tc_1+(1-t)c_2}
\end{eqnarray*}
Thus the vector $\frac{t\vec{x} + (1-t)\vec{y}}{tc_1+(1-t)c_2}$ is the row vector corresponding to $\omega$ of the matrix representation of $\mech^*$ and is therefore in $\rowcone(\priv)$. Multiplying by the nonnegative constant $tc_1+(1-t)c_2$, we get that $t\vec{x} + (1-t)\vec{y}\in\rowcone(\priv)$ and so $\rowcone(\priv)$ is convex.
\end{proof}

\section{Proof of Theorem \lowercase{\ref{thm:invcnfinf}}}\label{app:invcnfinf}

\begin{theorem}\emph{(Restatement and proof of Theorem \ref{thm:invcnfinf}).}
Let $\inp$ be a finite or countably infinite set of possible datasets. Let $\mech^*$ be an algorithm with $\domain(\mech^*)= \inp$. Let $M^*$ be the matrix representation of $\mech^*$ (Definition \ref{def:matrix}). If $(M^*)^{-1}$ exists and the $L_1$ norm of each column of $(M^*)^{-1}$ is bounded by a constant $C$ then 
\begin{list}{\labelitemi}{\leftmargin=1em}
\itemsep 1pt
\parskip 4pt
\item[(1)] A bounded row vector $\vec{x}\in\rowcone(\set{\mech^*})$ if and only if $\vec{x}\cdot m\geq 0$ for every column $m$ of $(M^*)^{-1}$.
\item[(2)] An algorithm $\mech$, with matrix representation $M$, belongs to $\cnf(\set{\mech^*})$ if and only if the matrix $M(M^*)^{-1}$ contains no negative entries.
\item[(3)] An algorithm $\mech$, with matrix representation $M$, belongs to $\cnf(\set{\mech^*})$ if and only if every row of $M$ belongs to $\rowcone(\set{\mech^*})$.
\end{list}
\end{theorem}
\begin{proof}
We first prove $(1)$. If $\vec{x}$ is the $0$ vector then this is clearly true. Thus assume $\vec{x}\neq\vec{0}$. If $\vec{x}\in\rowcone(\set{\mech^*})$ then by definition of the row cone and by Corollary \ref{cor:one}, $\vec{x}=\vec{y} M^*$ where $\vec{y}$ is a bounded row vector and has nonnegative components. Then $\vec{x}(M^*)^{-1}=\vec{y}M^*(M^*)^{-1}=\vec{y}$ and so $\vec{x}\cdot m\geq 0$ for every column $m$ of $(M^*)^{-1}$. 

For the other direction,  we must construct an algorithm $\randalg$ with matrix representation $A$ such that for some $c>0$, $c\vec{x}$ is a row of $AM^*$ (by definition of row cone and Corollary \ref{cor:one}). Thus, by hypothesis, suppose $\vec{x}\cdot m\geq 0$ for each column vector $m$ of $(M^*)^{-1}$ and consider the row vector $\vec{y}=\vec{x}(M^*)^{-1}$ which therefore has nonnegative entries. Since $\vec{x}$ is bounded and $||m||_1\leq C$ for each column vector $m$ of $(M^*)^{-1}$ then $|\vec{x}\cdot m|\leq ||\vec{x}||_\infty ||m||_1\leq ||\vec{x}||_\infty C$ (by H\"{o}lder's Inequality \cite{rudin}) so that $\vec{y}$ is bounded. Choose a $c$ so that $c\vec{y}$ is bounded by $1$. Consider the algorithm $\randalg$ that has a matrix representation $A$ with two rows, the first row being $c\vec{y}$ and the second row being $1-c\vec{y}$ ($\randalg$ is an algorithm since $c\vec{y}$ and $1-c\vec{y}$ have nonnegative components and the column sums of $A$ are clearly $1$). $\randalg$ is the desired algorithm since $c\vec{x}$ is a row of $AM^*$.

To prove (2) and (3), note that if an algorithm has matrix representation $M$, then $M(M^*)^{-1}$ contains all the dot products between rows of $M$ and columns of $(M^*)^{-1}$. Therefore, the entries of $M(M^*)^{-1}$ are nonnegative if and only if every row of $M$ is in the $\rowcone(\set{\mech^*})$ (this follows directly from the first part of the theorem). Thus (2) and (3) are equivalent and therefore we only need to prove (2).

To prove (2), first note the trivial direction. If $\mech\in\cnf(\set{\mech^*})$ then by definition every row of $M$ is in the row cone (and so by (1) all entries of $M(M^*)^{-1}$ are nonnegative). For the other direction, let $A=M(M^*)^{-1}$ (which has no negative entries by hypothesis). If we can show that the column sums of $A$ are all $1$ then, since $A$ contains no negative entries, $A$ would be a column stochastic matrix and therefore it would be the matrix representation of some algorithm $\randalg$. From this it would follow that $AM^*=M$ and therefore $\randalg\circ\mech^*=\mech$ (in which case $\mech\in\cnf(\mech^*)$ by Theorem \ref{thm:closure}). 

So all we need to do is to prove that the column sums of $A$ are all $1$. Let $\vec{1}$ be a column vector whose components are all $1$. Then since $M$ is a matrix representation of an algorithm (Definition \ref{def:matrix}), $M$ has column sums equal to $1$, and similarly for $M^*$. Thus:
\begin{eqnarray*}
\vec{1}^T&=&\vec{1}^T M^*(M^*)^{-1}\\
&=&\vec{1}^T(M^*)^{-1}\\
&&\text{and therefore}\\
\vec{1}^T A &=& \vec{1}^TM(M^*)^{-1}\\
&=&\vec{1}^T(M^*)^{-1}\\
&=&\vec{1}^T
\end{eqnarray*}
and so the column sums of $A$ are equal to $1$. This completes the proof of this theorem.

\end{proof}

%%%%%%%%%%%% END EAT

%--------------------------------------------------------------

\section{Proof of Lemma \lowercase{\ref{lem:phalf}}}\label{app:phalf}
\begin{lemma}\emph{(Restatement and proof of Lemma \ref{lem:phalf}).}
Given a privacy parameter $p$, define $q=\max(p,1-p)$. Then
\begin{list}{\labelitemi}{\leftmargin=2em}
\itemsep 0pt
\parskip 2pt
\item $\cnf(\set{\rr{p}})=\cnf(\set{\rr{q}})$.
\item  If $p=1/2$ then $\cnf(\set{\rr{p}})$ consists of the set of algorithms whose outputs are statistically independent of their inputs (i.e. those algorithms $\mech$ where $P[\mech(D_i)=\omega]=P[\mech(D_j)=\omega]$ for all $D_i,D_j\in\inp$ and $\omega\in\range(\mech)$), and therefore attackers learn nothing from those outputs.
\end{list}
\end{lemma}
\begin{proof}
Consider the algorithm $\rr{0}$ which always flips each bit in its input. It is easy to see that $\rr{0}\circ\rr{p}=\rr{1-p}$ and  $\rr{0}\circ\rr{1-p}=\rr{p}$. From Theorem \ref{thm:closure}, it follows that $\cnf(\set{\rr{p}})=\cnf(\set{\rr{1-p}})$ and therefore  $\cnf(\set{\rr{p}})=\cnf(\set{\rr{q}})$.

Clearly, the output of $\rr{1/2}$ is independent of whatever was the true input table $D\in\inp$. By Theorem \ref{thm:closure}, all algorithms in $\cnf(\set{\rr{1/2}})$ have outputs independent of their inputs. For the other direction, choose any algorithm $\mech$ whose outputs are statistically independent of their inputs. Then it is easy to see that $\mech=\mech\circ\rr{1/2}$; that is, $\mech$ and $\mech\circ\rr{1/2}$ have the same range and $P(\mech(D_i)=\omega)=P([\mech\circ\rr{1/2}](D_i)=0)$ for all $D_i\in\inp$ and $\omega\in\range(\mech)$. Thus $\mech\in\cnf(\set{\rr{1/2}})$. 

Clearly, when the output is statistically independent of the input, an attacker can learn nothing about the input after observing the output.
\end{proof}

%-------------------------------------------------------

\section{Proof of Theorem \lowercase{\ref{thm:rrcnf}}}\label{app:rrcnf}
\begin{theorem}\emph{(Restatement and proof of Theorem \ref{thm:rrcnf}).}
Given input space $\inp=\set{D_1,\dots,D_{2^k}}$ of bit strings of length $k$ and a privacy parameter $p> 1/2$,
\begin{list}{\labelitemi}{\leftmargin=1em}
\itemsep 4pt
\parskip 2pt
\item A vector $\vec{x}=(x_1,\dots,x_{2^k})\in\rowcone(\set{\rr{p}})$ if and only if for every bit string $s$ of length $k$,
$$\sum\limits_{i=1}^{2^k}p^{\hamming(s, D_i)}(p-1)^{k-\hamming(s, D_i)}x_i \geq 0$$
where $\hamming(s, D_i)$ is the Hamming distance between $s$ and $D_i$.
\item An algorithm $\mech$ with matrix representation $M$ belongs to $\cnf(\set{\rr{p}})$ if and only if every row of $M$ belongs to $\rowcone(\set{\rr{p}})$.
\end{list}
\end{theorem}
\begin{proof}
Our strategy is to first derive the matrix representation of $\rr{p}$, which we denote by $\mrr{p}$. Then we find the inverse of $\mrr{p}$ and apply Theorem \ref{thm:invcnfinf}. Accordingly, we break the proof down into 3 steps.

\noindent\textbf{Step 1:} \ul{Derive $\mrr{p}$}. Define $B$ to be the matrix
\begin{eqnarray*}
B&=&
\left(
\begin{matrix}
p & 1-p\\
1-p & p
\end{matrix}
\right)
\end{eqnarray*}
Recall that the Kronecker product $C\oplus D$ of an $m\times n$ matrix $C$ and $m^\prime\times n^\prime$ matrix $D$ is the block matrix 
$\left(
\begin{smallmatrix}
c_{11}D & \dots & c_{1n}D\\ 
\vdots & \ddots & \vdots\\ 
c_{m1}D & \dots &c_{mn}D
\end{smallmatrix}
\right)$ 
of dimension $mm^\prime\times nn^\prime$. An easy induction shows that the matrix representation $\mrr{p}$ is equal to the $k$-fold Kronecker product of $B$ with itself:
$$\mrr{p}=\bigotimes\limits_{i=1}^k B$$

The entry in row $i$ and column $j$ of $\mrr{p}$ is equal to $P[\rr{P}(D_j)=D_i]$ and a direct computation shows that this is equal to 
$$p^{\hamming(D_i,D_j)}(1-p)^{k-\hamming(D_i,D_j)}$$

\noindent\textbf{Step 2:} \ul{Derive $(\mrr{p})^{-1}$}.
It is easy to check that 
\begin{eqnarray*}
B^{-1}&=&
\frac{1}{2p-1}\left(
\begin{matrix}
p & -(1-p)\\
-(1-p) & p
\end{matrix}
\right)
\end{eqnarray*}
and therefore 
$$(\mrr{p})^{-1}=\bigotimes\limits_{i=1}^k B^{-1}$$
A comparison with $\bigotimes\limits_{i=1}^k B^{-1}$ shows that the we can calculate the entry in row $i$ and column $j$ of $(\mrr{p})^{-1}$ by taking the corresponding entry of $\mrr{p}$ and replacing every occurrence of $1-p$ with $-(1-p)=p-1$. Thus the entry in row $i$ and column $j$ of $(\mrr{p})^{-1}$ is equal to
$$\frac{1}{(2p-1)^k}p^{\hamming(D_i,D_j)}(p-1)^{k-\hamming(D_i,D_j)}$$
Therefore each column of $(\mrr{p})^{-1}$ has the form:
\begin{eqnarray*}
\frac{1}{(2p-1)^k}
\begin{bmatrix}
p^{\hamming(s,D_1)}(p-1)^{k-\hamming(s,D_1)}\\
p^{\hamming(s,D_2)}(p-1)^{k-\hamming(s,D_2)}\\
\vdots\\
p^{\hamming(s,D_{2^k})}(p-1)^{k-\hamming(s,D_{2^k})}\\
\end{bmatrix}
\end{eqnarray*}
\noindent\textbf{Step 3:} Now we apply Theorem \ref{thm:invcnfinf} and observe that if $m^{(i)}$ is the $i^\text{th}$ column of $(\mrr{p})^{-1}$, then, since $p>1/2$ and $2p-1>0$, the condition $\vec{x}\cdot m^{(i)}$ is equal to the condition
$$\sum\limits_{j=1}^{2^k}p^{\hamming(s, D_j)}(p-1)^{k-\hamming(s, D_j)}x_j \geq 0$$
where $s=D_i$.
\end{proof}

%----------------------------------------------------------

\section{Proof of Theorem \lowercase{\ref{thm:rrsemantics}}}\label{app:rrsemantics}

\begin{theorem}\emph{(Restatement and proof of Theorem \ref{thm:rrsemantics}).}
Let $p$ be a privacy parameter and let $\inp={D_1,\dots, D_{2^k}}$. Let $\mech$ be an algorithm that has a matrix representation whose every row  belongs to the row cone of randomized response. If the attacker believes that the bits in the data are independent and bit $i$ is equal to $1$ with probability $q_i$, then $\mech$ protects the parity of any subset of bits that have prior probability $\geq p$ or $\leq 1-p$. That is, for any subset $\set{\ell_1,\dots,\ell_m}$ of bits of the input data such that  $q_{\ell_j}\geq p~\vee~g_{\ell_j} \leq 1-p$ for $j=1,\dots, m$, the following holds:
\begin{list}{\labelitemi}{\leftmargin=0.5em}
\itemsep 4pt
\parskip 2pt
\item If $P(\parity(J)=0) \geq P(\parity(J)=1)$ then $P(\parity(J)=0~|~\mech(\data)) \geq P(\parity(J)=1~|~\mech(\data))$ 
\item If $P(\parity(J)=1) \geq P(\parity(J)=0)$ then $P(\parity(J)=1~|~\mech(\data)) \geq P(\parity(J)=0~|~\mech(\data))$ 
\end{list}
Furthermore, an algorithm $\mech$ can only provide these guarantees if every row of its matrix representation belongs to $\rowcone(\set{\rr{p}})$.
\end{theorem}
\begin{proof}
We break this proof up into a series of steps. We first reformulate the statements to make them easier to analyze mathematically, then we specialize to the case where $J=\set{1,\dots,k}$ is the set of all bits in the database. We then show that every $\mech$ whose rows (in the corresponding matrix representation) belong to $\rowcone({\rr{p}})$ has these semantic guarantees. We then show that only those $\mech$ provide these semantic guarantees. Finally we show that those results imply that the theorem holds for all $J$ whose bits have prior probability $\geq p$ or $\leq 1-p$.

\vspace{0.5em}
\noindent\textbf{Step 1:} \ul{Problem reformulation and specialization to the case when $J=\set{1,\dots,k}$}. 
Assume $J=\set{1,\dots,k}$ so that for all bits $j$, either $q_j\geq p$ or $q_j\leq 1-p$.

First, Lemma \ref{lem:phalf} allows us to assume that the privacy parameter \ul{$p>1/2$ without any loss of generality}: the case of $p=1/2$ is trivial since the output provides no information about the input so that parity is preserved; in the case of $p<1/2$, the row cone and $\cnf$ are unchanged if we replace $p$ with $1-p$.

Second, we need a few results about parity. An easy induction shows that:
\begin{eqnarray*}
P(\parity(\data)=1)&=&\frac{1-\prod\limits_{j=1}^k (1-2q_j)}{2}\\
P(\parity(\data)=0)&=&\frac{1+\prod\limits_{j=1}^k(1-2q_j)}{2}
\end{eqnarray*}
in particular, if all of the $q_j\neq 1/2$ then $P(\parity(\data)=1)\neq P(\parity(\data)=0)$ so that one parity has higher prior probability than the other.

When $J$ is the set of all $k$ bits, then for all $q_j$, $q_j\neq 1/2$ and so the parities cannot be equally likely \textit{a priori}, the statement about protection of parity can be rephrased as $P(\parity(\data)=0) - P(\parity(\data)=1)$ and $P(\parity(\data)=0~|~\mech(\data)) - P(\parity(\data)=1~|~\mech(\data))$ have the same sign or the posterior probabilities of parity are the same. Equivalently, 
\begin{eqnarray}
0&\leq&\Big(P[\parity(\data)=0] - P[\parity(\data)=1]\Big) \nonumber \\
&& \times \Big(P[\parity(\data)=0~|~\mech(\data)] - P[\parity(\data)=1~|~\mech(\data)]\Big)\label{eqn:prodparity}
\end{eqnarray}

Now, it is easy to see that
\begin{eqnarray}
 \lefteqn{P(\parity(\data)=0) - P(\parity(\data)=1)}\nonumber\\
&=& \left[\bigotimes\limits_{j=1}^k (-q_{j},~ 1-q_{j})\right]\cdot \left[\bigotimes\limits_{j=1}^k (1, 1)\right]\nonumber\\
&=& \prod\limits_{j=1}^k\Big[(-q_j,~ 1-q_{j})\cdot (1,1)\Big]\label{eqn:priorsub}
\end{eqnarray}
and
\begin{eqnarray}
\lefteqn{\hspace{-2cm}P[\parity(\data)=0~|~\mech(\data)] - P[\parity(\data)=1~|~\mech(\data)]}\nonumber\\
&&=\alpha  \left[\bigotimes\limits_{j=1}^k (-q_{j}, 1-q_{j})\right] \cdot\vec{x}\label{eqn:postsub}
\end{eqnarray}
where $\alpha$ is a positive normalizing constant and $\vec{x}$ is a vector of the matrix
representation of $\mech$. So, by Equations \ref{eqn:prodparity}, \ref{eqn:priorsub}, and \ref{eqn:postsub},
the statement about protecting parity is equivalent to
\begin{eqnarray}
\forall \vec{x}\in\rowcone(\set{\rr{p}}) ~:~ 0 \leq \left(\prod\limits_{j=1}^k\Big[(-q_{j}, ~1-q_{j})\cdot (1,1)\Big]\right) *\left(\left[\bigotimes\limits_{j=1}^k (-q_{j}, ~1-q_{j})\right] \cdot\vec{x}\right)\label{eqn:toprove}
\end{eqnarray}

\vspace{0.5em}
\noindent\textbf{Step 2:} \ul{Show that if for all $j$, $q_j\geq p\vee q_j\leq 1-p$ then the constraints in Equation }\ref{eqn:toprove}\ul{ hold (i.e. the most likely parity \textit{a priori} is the most likely parity \textit{a posteriori})}.

It follows from % Theorem \ref{thm:closure}
Corollary \ref{cor:one}
 that every $\mech\in\cnf(\set{\rr{p}})$ has the form $\randalg\circ\rr{p}$ and so, by Theorem \ref{thm:invcnfinf},  $\vec{x}$ is a row from the matrix representation of an $\mech\in\cnf(\set{\rr{p}})$ if and only if $\vec{x}\in\rowcone(\set{\rr{p}})$. This means that ever such $\vec{x}$ is a nonnegative
linear combination of rows of the randomized response algorithm $\rr{p}$. Thus it suffices to show that 
\begin{eqnarray}
0&\leq& \left(\prod\limits_{j=1}^k\Big[(-q_{j}, ~1-q_{j})\cdot (1,1)\Big]\right) * \left(\left[\bigotimes\limits_{j=1}^k (-q_{j}, ~1-q_{j})\right] \cdot\vec{m}\right)\label{eqn:toprove2}
\end{eqnarray}
for each vector $\vec{m}$ in $\mrr{p}$ (the matrix representation of $\rr{p}$).
It is easy to check that
\begin{eqnarray*}
\mrr{p}=\bigotimes\limits_{i=1}^k
\begin{pmatrix}
p & 1-p\\
1-p & p
\end{pmatrix}
\end{eqnarray*}
and so every vector $\vec{m}$ that is a row of $\mrr{p}$ has the form $\bigotimes\limits_{i=1}^k v_i$ where $v_i=(p,~1-p)$ or $(1-p, ~p)$. Thus right hand side of Equation \ref{eqn:toprove2} has the form:

\begin{eqnarray}
\prod\limits_{j=1}^k\Big[\Big((-q_{j}, ~1-q_{j})\cdot (1,1)\Big)*\Big((-q_j,~1-q_j)\cdot v_i\Big)\Big]\label{eqn:toprove3}
\end{eqnarray}
where $v_i=(p,~1-p)$ or $(1-p, ~p)$. Each term in this product is either 
$$(1-2q_j)*[(1-p)(1-q_j)-q_jp]=(1-2q_j)[1-p-q_j]$$
or 
$$(1-2q_j)*[p(1-q_j)-q_j(1-p)]=(1-2q_j)[p-q_j]$$ 
Recalling that we had assumed $p>1/2$ without any loss of generality, both of these terms are nonnegative if $q_j\geq p > 1/2$ and they are also nonnegative when $q_i\leq (1-p) < 1/2$. Thus the product in Equation \ref{eqn:toprove3} is nonnegative from which it follows that the conditions in Equation \ref{eqn:toprove2} and \ref{eqn:toprove} are satisfied which implies Equation \ref{eqn:prodparity} is satisfied, which proves half of the theorem when restricted to the special case of $J=\set{1,\dots,k}$.

\vspace{0.5em}
\noindent\textbf{Step 3:} \ul{Show that if $\mech$ is a mechanism that protects parity whenever $q_j\geq p~\vee~ q_j\leq 1-p$ for $i=1,\dots,k$ then every row $\vec{x}$ in its matrix representation belongs to $\rowcone(\set{\rr{p}})$}.

We actually prove a more general statement: if $\mech$ is a mechanism that protects parity whenever $q_j= p~\vee~ q_j= 1-p$ for $i=1,\dots,k$ then every row $\vec{x}$ in its matrix representation belongs to $\rowcone(\set{\rr{p}})$.

Recalling the argument leading up to Equation \ref{eqn:toprove} in Step 2 (where we reformulated the problem into a statement that is more amenable to mathematical manipulation), we need to show that if
\begin{eqnarray}
0\leq \left(\prod\limits_{j=1}^k\Big[(-q_{j}, ~1-q_{j})\cdot (1,1)\Big]\right) *\left(\left[\bigotimes\limits_{j=1}^k (-q_{j}, ~1-q_{j})\right] \cdot\vec{x}\right)\label{eqn:step3toprove}
\end{eqnarray}
whenever $q_j=p$ or $q_j=1-p$ then $\vec{x}\in\rowcone(\set{\rr{p}})$.

Define the function:
\begin{eqnarray*}
\sign(\alpha)=\begin{cases}
-1 & \text{ if } \alpha < 0\\
0 & \text{ if } \alpha = 0\\
1 & \text{ if } \alpha > 0\\
\end{cases}
\end{eqnarray*}

Simplifying Equation \ref{eqn:step3toprove} (by computing the dot product in the first term, looking just at the sign of that dot product, and then combining both terms), our goal is to show that if
\begin{eqnarray}
0&\leq& 
\left(\left[\bigotimes\limits_{j=1}^k (-q_{j}, ~1-q_{j})*\sign(1-2q_j)\right] \cdot\vec{x}\right)\label{eqn:step3toprove2}
\end{eqnarray}
whenever $q_j=p$ or $q_j=1-p$ then $\vec{x}\in\rowcone(\set{\rr{p}})$.

Now, when $q_j=p$ (and recalling that we have assumed $p>1/2$ with no loss of generality in Step 1), then $$(-q_{j}, ~1-q_{j})*\sign(1-2q_j)=(p,~ -(1-p))$$ and when $q_j=1-p$ then $$(-q_j, ~1-q_j)*\sign(1-2q_j)=(-(1-p), ~p)$$

Thus asserting that Equation \ref{eqn:step3toprove2} holds whenever $q_j$ equals $p$ or $1-p$ is the same as asserting that the vector:
\begin{eqnarray}
\vec{x}^T \bigoplus\limits_{i=1}^k \frac{1}{2p-1}
\begin{pmatrix}
p & -(1-p)\\
-(1-p) & p
\end{pmatrix}\label{eqn:step3toprove3}
\end{eqnarray}
has no negative components. However, the randomized response algorithm  $\rr{p}$ has a matrix representation $\mrr{p}$ whose inverse (which we also derived in the proof of Theorem \ref{thm:rrcnf}) is 
\begin{eqnarray*}
(\mrr{p})^{-1}=\bigoplus\limits_{i=1}^k \frac{1}{2p-1}
\begin{pmatrix}
p & -(1-p)\\
-(1-p) & p
\end{pmatrix}
\end{eqnarray*}

Thus the condition that the vector in Equation \ref{eqn:step3toprove3} has no negative entries means that $\vec{x}^T(\mrr{p})^{-1}$ has no negative entries and so the dot product of $\vec{x}$ with any column of $(\mrr{p})^{-1}$ is nonnegative. By Theorem \ref{thm:rrcnf}, this means that $\vec{x}\in\rowcone(\set{\rr{p}})$.

This concludes the proof for the entire theorem specialized to the case where $J=\set{1,\dots,k}$. In the next step, we generalize this to arbitrary $J$.

\vspace{0.5em}
\noindent\textbf{Step 4:} 
Now let $J=\set{\ell_1,\dots,\ell_m}$. First consider an ``extreme'' attacker whose prior beliefs $q_j$ are such that $q_j=0$ or $q_j=1$ whenever $j\notin J$. It follows from the previous steps that such an attacker would not change his mind about the parity of the whole dataset. Since the attacker is completely sure about the values of bits outside of $J$, this means that after seeing a sanitized output $\omega$, the attacker will not change his mind about the parity of the bits in $J$.

Now, note that showing
\begin{list}{\labelitemi}{\leftmargin=0.5em}
\itemsep 4pt
\parskip 2pt
\item If $P(\parity(J)=0) \geq P(\parity(J)=1)$ then\\  $P(\parity(J)=0~|~\mech(\data)=\omega) \geq P(\parity(J)=1~|~\mech(\data)=\omega)$ 
\item If $P(\parity(J)=1) \geq P(\parity(J)=0)$ then\\  $P(\parity(J)=1~|~\mech(\data)=\omega) \geq P(\parity(J)=0~|~\mech(\data)=\omega)$ 
\end{list}
is equivalent to showing
\begin{list}{\labelitemi}{\leftmargin=0.5em}
\itemsep 4pt
\parskip 2pt
\item If $P(\parity(J)=0) \geq P(\parity(J)=1)$ then\\ $P(\parity(J)=0~\wedge~\mech(\data))\geq P(\parity(J)=1~\wedge~\mech(\data))$ 
\item If $P(\parity(J)=1) \geq P(\parity(J)=0)$ then\\ $P(\parity(J)=1~\wedge~\mech(\data)=\omega) \geq P(\parity(J)=0~\wedge~\mech(\data)=\omega)$ 
\end{list}
since we just multiply the equations on both sides of the inequalities by the positive number $P(\mech(\data)=\omega)$.

Now consider an attacker Bob such that $q_j\geq p$ or $q_j\leq 1-p$ whenever $j\in J$ and there are no restrictions on $q_j$ for $j\notin J$. There is a corresponding set of $2^{k-|J|}$ ``extreme'' attackers for whom $P($bit $j=1)=q_j$ for $j\in J$ and $P($bit $j=1)\in\set{0,1}$ otherwise.

Bob's vector of prior probabilities over possible datasets
$$(P[\data=D_1], P[\data=D_2], \dots)$$ 
is a convex combination of the corresponding vectors for the extreme attackers.
and thus Bob's joint distributions:
$$P(\parity(J)=1~\wedge~\mech(\data)=\omega)$$
and
$$P(\parity(J)=0~\wedge~\mech(\data)=\omega)$$
are convex combinations of the corresponding posteriors for the extreme attackers, and the coefficients of this convex combination are the same. 

Note that Bob and all of the extreme attackers have the same prior on the parity of $J$. However, we have shown that the extreme attackers will not change their minds about the parity of $J$. Therefore if they believe $P(\parity(J)=1~\wedge~\mech(\data)=\omega)$ is larger than the corresponding probability for even parity, then Bob will have the same belief. If the extreme attackers believe, after seeing the sanitized output $\omega$, that even parity is more likely, then so will Bob. Thus Bob will not change his belief about the parity of the input dataset.
\end{proof}

%--------------------------------------------------------

\section{Proof of Lemma \lowercase{\ref{lem:frappapprox}}}\label{app:frappapprox}

\begin{lemma}\emph{(Restatement and proof of Lemma \ref{lem:frappapprox}).}
Let $p=\frac{\gamma}{\gamma+1}$. Then $\tilde{K}_p$ is an approximation cone for $\frapp$.
\end{lemma}
\begin{proof}
Clearly $\tilde{K}_p$ is a closed convex cone. Thus we just need to prove that $\rowcone(\frapp)\subseteq \tilde{K}_p$.

Choose any $\mech_Q\in\frapp$, with matrix representation $M_Q$.
Clearly 
$$M_Q=\bigotimes\limits_{i=1}^k Q$$
and $Q$ satisfies the constraints
$$\forall i,j\in\set{1,\dots, N}~:~ Q (p e_i - (1-p) e_j)\succeq \vec{0}$$
where  $e_{i}$ is the $i^\text{th}$ column vector of the $N\times N$ identity matrix and $\vec{a}\succeq\vec{b}$ means that $\vec{a}-\vec{b}$ has no negative components. It follows from the properties of the Kronecker product that 
\begin{eqnarray}
\forall i_1,\dots, i_k,j_1,\dots,j_k\in\set{1,\dots,N}~:~ M_Q\left( \bigotimes\limits_{\ell=1}^k (p e_{i_\ell} - (1-p)e_{j_\ell})\right)\succeq \vec{0} \label{eqn:frappconstraints}
\end{eqnarray}
Thus each row of the matrix representation of $\mech_Q$ satisfies a set of linear constraints.

From Theorem \ref{thm:closure}, we see that $\cnf(\frapp)$ can be obtained by first creating all algorithms of the form $\randalg\circ\mech_Q$ (for $\mech_Q\in\frapp$) and then by taking the convex combination of those results (i.e. creating algorithms that randomly choose to run one of the algorithms generated in the previous step). However,  the matrix representation of $\randalg\circ\mech_Q$ is equal to $AM_Q$ (where $A$ is the matrix representation of $\randalg$) and every row in $AM_Q$ is a positive linear combination of rows in $\mech_Q$. Thus every row of the matrix representation of $\randalg\circ\mech_Q$ also satisfies the constraints defining $\tilde{K}_p$. Finally, creating an algorithm $\randalg^*$ that randomly choose to run one algorithm in $\set{\randalg_1\circ\mech_{Q_1},\dots,\randalg_h\circ\mech_{Q_h}}$ means that the rows in the matrix representation of $\randalg^*$ is a convex combination of the rows appearing in the matrix representations of the $\randalg_i\circ\mech_{Q_i}$ and so those rows also satisfy the constraints that define $\tilde{K}_p$. Therefore $\rowcone(\frapp)\subseteq \tilde{K}_p$.
\end{proof}

%---------------------------------------------------------

\section{Proof of Theorem \lowercase{\ref{thm:skellam}}}\label{app:skellam}

\begin{theorem}(\emph{Restatement and proof of Theorem \ref{thm:skellam}})
Let the input domain $\inp=\set{\dots, -2, -1, 0,\linebreak[0] 1, 2, \dots}$  be the set of integers. Let $\mech_{\text{skell($\lambda_1,\lambda_2$)}}$ be the algorithm that adds to its input a random integer $k$ with the Skellam$(\lambda_1,\lambda_2)$ distribution and let $f_Z(\cdot; \lambda_1,\lambda_2)$ be the probability mass function of the Skellam$(\lambda_1,\lambda_2)$ distribution. A  bounded row vector $\vec{x}=(\dots, x_{-2}, x_{-1}, x_0, x_1, x_2, \dots)$ belongs to $\rowcone(\set{\mech_{\text{skell($\lambda_1,\lambda_2$)}}})$ if for all integers $k$,
\begin{eqnarray*}
\sum\limits_{j=-\infty}^\infty (-1)^j f_Z(j;\lambda_1,\lambda_2)x_{k+j}\geq 0
\end{eqnarray*}
\end{theorem}
\begin{proof}
For integers $k$, define the functions
\begin{eqnarray*}
f_X(k)&=&
\begin{cases}
e^{-\lambda_1} \frac{\lambda_1^k}{k!}&\text{ if } k\geq 0\\
0 &\text{ if } k<0
\end{cases}\\
f_Y(k)&=&
\begin{cases}
e^{-\lambda_2} \frac{\lambda_2^{-k}}{(-k)!}&\text{ if } k\leq 0\\
0 &\text{ if } k>0
\end{cases}
\end{eqnarray*}
Note that $f_X$ is the probability mass function for a Poisson$(\lambda_1)$ random variable $X$ while \ul{$f_Y$ is the probability mass function of the \textbf{negative} of a Poisson$(\lambda_2)$ random variable $Y$}. 

With this notation, the Skellam distribution is the distribution of the sum $X+Y$. Therefore its probability mass function satisfies the following relation 
$$f_Z(k; \lambda_1,\lambda_2)=\sum\limits_{j=-\infty}^\infty f_X(k-j)f_Y(j)=(f_X\star f_Y)(k)$$
where $f_X\star f_Y$ is the convolution operation.

Now for each integer $k$ define
\begin{eqnarray*}
g_X(k)&=& (-1)^k f_X(k)\\
g_Y(k)&=&(-1)^k f_Y(k)\\
g_Z(k)&=&(g_X\star g_Y)(k)\\
&=&\sum\limits_{j=-\infty}^\infty g_X(k-j)g_Y(j)\\
&=&\sum\limits_{j=-\infty}^\infty (-1)^{k-j}f_X(k-j)(-1)^jf_Y(j)\\
&=&(-1)^k\sum\limits_{j=-\infty}^\infty f_X(k-j)f_Y(j)\\
&=&(-1)^k f_Z(k;\lambda_1,\lambda_2)
\end{eqnarray*}

We will need the following calculations:
\begin{eqnarray*}
(g_X\star f_X)(k)&=&\sum\limits_{j=-\infty}^\infty g_X(k-j)f_X(j)\\
&=&\sum\limits_{j=-\infty}^\infty (-1)^{k-j}f_X(k-j)f_X(j)\\
&=&\sum\limits_{j=0}^k (-1)^{k-j}f_X(k-j)f_X(j)\\
&&\text{(since $f_X$ is $0$ for negative integers also note the summation is $0$ if $j>k)$}\\
&=&\ind{k\geq 0}e^{-2\lambda_1}\sum\limits_{j=0}^k \frac{(-\lambda_1)^{k-j}}{(k-j)!}\frac{\lambda_1^j}{j!}\\
&=&\ind{k\geq 0}\frac{e^{-2\lambda_1}}{k!}\sum\limits_{j=0}^k {k\choose j}(-\lambda_1)^{k-j}\lambda_1^j\\
&=&\ind{k\geq 0}\frac{e^{-2\lambda_1}}{k!}(\lambda_1-\lambda_1)^k\\
&=&\begin{cases}e^{-2\lambda_1}\text{ if }k=0\\0 \text{ otherwise}\end{cases}
\end{eqnarray*}
Similarly
\begin{eqnarray*}
(g_Y\star f_Y)(k)&=&\sum\limits_{j=-\infty}^\infty g_Y(k-j)f_Y(j)\\
&=&\sum\limits_{j=-\infty}^\infty (-1)^{k-j}f_Y(k-j)f_Y(j)\\
&=&\sum\limits_{j=k}^0 (-1)^{k-j}f_Y(k-j)f_Y(j)\\
&&\text{(since $f_Y$ is $0$ for positive integers also note the summation is $0$ if $k>j)$}\\
&=&\ind{k\leq 0}e^{-2\lambda_2}\sum\limits_{j=k}^0 \frac{(-\lambda_2)^{-(k-j)}}{(-(k-j))!}\frac{\lambda_2^{-j}}{(-j)!}\\
&=&\ind{k\leq 0}e^{-2\lambda_2}\sum\limits_{j=0}^{-k} \frac{(-\lambda_2)^{(-k)-j}}{[(-k)-j]!}\frac{\lambda_2^{j}}{j!}\\
&&\text{(replacing the dummy index $j$ with $-j$)}\\
&=&\ind{k\leq 0}\frac{e^{-2\lambda_2}}{(-k)!}\sum\limits_{j=0}^{-k} {(-k)\choose j}(-\lambda_2)^{(-k)-j}\lambda_2^j\\
&=&\ind{k\leq 0}\frac{e^{-2\lambda_2}}{(-k)!}(\lambda_2-\lambda_2)^{(-k)}\\
&=&\begin{cases}e^{-2\lambda_2}\text{ if }k=0\\0 \text{ otherwise}\end{cases}
\end{eqnarray*}
From these calculations we can conclude that 
\begin{eqnarray*}
(g_Z\star f_Z(\cdot;\lambda_1,\lambda_2))(k)&=&((g_X\star g_Y)\star(f_X\star f_Y))(k)\\
&=&((g_X\star f_X)\star(g_Y\star f_Y))(k)\\
&&\text{(since convolutions are commutative and associative)}\\
&=&\begin{cases}
e^{-2(\lambda_1+\lambda_2)} & \text{  if }k=0\\
0 & \text{ otherwise}
\end{cases}
\end{eqnarray*}
These convolution calculations show that the  matrices $M^{(f)}$ and $M^{(g)}$, whose rows and columns are indexed by the integers and which are defined below, are inverses of each other.
\begin{eqnarray*}
M^{(f)}_{(i,j)}&\equiv&(i,j) \text{ entry of }M^{(f)} \\
&=&f_Z(i-j;\lambda_1,\lambda_2)\\
M^{(g)}_{(i,j)}&\equiv&(i,j)\text{ entry of }M^{(g)}\\
 &=& e^{2(\lambda_1+\lambda_2)}g_Z(i-j) 
\end{eqnarray*}
To see that they are inverses, note that the dot product between row $r$ of $M^{(f)}$ and column $c$ of $M^{(g)}$ is
\begin{eqnarray*}
\sum\limits_{j=-\infty}^\infty M^{(f)}_{(r,j)}M^{(g)}_{(j,c)}&=&\sum\limits_{j=-\infty}^\infty f_Z(r-j;\lambda_1,\lambda_2)e^{2(\lambda_1+\lambda_2)}g_Z(j-c)\\
&=&\sum\limits_{j=-\infty}^\infty f_Z(r-c-j;\lambda_1,\lambda_2)e^{2(\lambda_1+\lambda_2)}g_Z(j)\\
&=&e^{2(\lambda_1+\lambda_2)}(f_Z(\cdot;\lambda_1,\lambda_2)\star g_Z)(r-c)\\
&=&e^{2(\lambda_1+\lambda_2)}(g_Z\star f_Z(\cdot;\lambda_1,\lambda_2))(r-c)\\
&=&\begin{cases}
1 & \text{ if }r=c\\
0 & \text{ otherwise}
\end{cases}
\end{eqnarray*}
Now,  clearly $M^{(f)}$ is the matrix representation of $\mech_{\text{skell($\lambda_1,\lambda_2$)}}$ so that we can again use Theorem \ref{thm:invcnfinf} and the observation that $g_Z(k)=(-1)^kf_Z(k;\lambda_1,\lambda_2)$ so that column $c$ of $M^{(g)}=(M^{(f)})^{-1}$ is the column vector whose entry $j$ is $(-1)^{j-c}f_Z(j-c;\lambda_1,\lambda_2)$.

Note that the columns of $M^{(g)}$ have bounded $L_1$ norm since the absolute value of the entries in any column are proportional to the probabilities given by the Skellam distribution.

The proof is completed by the observation that for any $\vec{x}=(\dots,x_{-2}, x_{-1}, x_0, x_1, x_2, \dots)$,
\begin{eqnarray*}
\sum\limits_{j=-\infty}^\infty (-1)^{j-c}f_Z(j-c;\lambda_1,\lambda_2)x_j = \sum\limits_{j=-\infty}^\infty (-1)^{j}f_Z(j;\lambda_1,\lambda_2)x_{j+c}
\end{eqnarray*}
\end{proof}
\section{Proof of Lemma \lowercase{\ref{lem:mechgeo}}}\label{app:mechgeo}
\begin{lemma}\emph{(Proof and restatement of Lemma \ref{lem:mechgeo}).}
$\linebreak[0]\mech_{DNB(p,1)}$, the differenced negative binomial mechanism with $r=1$, is the geometric mechanism.
\end{lemma}
\begin{proof}
We need to show that the difference between two independent Geometric$(p)$ distributions has the probability mass function $f(k)=\frac{1-p}{1+p}p^{|k|}$.

Let $X$ and $Y$ be independent Geometric$(p)$ random variables and let $Z=X-Y$. Then
\begin{eqnarray*}
P(Z=k)&=&
\begin{cases}
\sum\limits_{j=0}^\infty P(X=j+k)P(Y=j) &\text{ if }k\geq 0\\
\sum\limits_{i=0}^\infty P(X=j)P(Y=j+|k|) &\text{ if }k<0
\end{cases}
\end{eqnarray*}
Combining both cases, we get
\begin{eqnarray*}
P(Z=k)&=&\sum\limits_{j=0}^\infty (1-p)p^{j+|k|}(1-p)p^j\\
         &=&(1-p)^2p^{|k|}\sum\limits_{j=0}^\infty(p^2)^j\\
         &=&(1-p)^2p^{|k|}\frac{1}{1-p^2}\\
         &=&(1-p)^2p^{|k|}\frac{1}{(1-p)(1+p)}\\
        &=&\frac{1-p}{1+p}p^{|k|}
\end{eqnarray*}
\end{proof}

%%%%%%%%%%%%%%%%%%%%%%%%%%%%%%%%%%%%%%%%%
%%% EAT THEOREM SPECIFIC TO GEOMETRIC MECHANISM
\eat{
%%%%%%%%%%%%%%%%%%%%%%%%%%%%%%%%%%%%%%%%%
\begin{theorem}\label{thm:geo}
Let the input domain $\inp=\set{\dots, -2, -1, 0, 1, 2, \dots}$ be the set of integers. Let $\mech_{\epsilon\text{-Geo}}$ be the algorithm that adds to its input a random integer $k$ with distribution $\frac{e^\epsilon-1}{e^\epsilon+1} e^{-\epsilon |k|}$. A  bounded row vector $\vec{x}=(\dots, x_{-2}, x_{-1}, x_0, x_1, x_2, \dots)$ belongs to $\rowcone(\set{\mech_{\epsilon\text{-Geo}}})$ if for all integers $k$,
$$-x_{k-1} + (e^\epsilon+1/e^\epsilon)x_k - x_{k+1}\geq 0$$
\end{theorem}
\begin{proof}
Let $M_{\epsilon\text{-Geo}}$ be the matrix representation of $\mech_{\epsilon\text{-Geo}}$.  Define the matrix $W$ whose rows and columns are indexed by the integers such that the $(i,j)^\text{th}$ entry of $W$, denoted by $W_{i,j}$ is
\begin{eqnarray*}
W_{i,j}&=&
\begin{cases}
-1 & \text{ if } |i-j|=1\\
e^\epsilon + 1/e^\epsilon & \text{ if } i=j\\
0 & \text{ otherwise }
\end{cases}
\end{eqnarray*}
it is easy to see that $(e^\epsilon/(e^\epsilon-1)^2)W$ is the inverse of $M_{\epsilon\text{-Geo}}$. Since each column of $W$ has $L_1$ norm equal to $e^\epsilon+1/e^\epsilon+2$, we can apply Theorem \ref{thm:invcnfinf} so that the row vector $\vec{x}$ is in the row cone if and only if $\vec{x}W$ has no negative components.
\end{proof}
%%%%%%%%%%%%%%%%%%%%%%%%%%%%%%%%%%%%%%%%%
%%% EAT THEOREM SPECIFIC TO GEOMETRIC MECHANISM
}
%%%%%%%%%%%%%%%%%%%%%%%%%%%%%%%%%%%%%%%%%

%----------------------------------------------------
\section{Proof of Theorem \lowercase{\ref{thm:dnbrowcone}}}\label{app:dnbrowcone}
We first need an intermediate result.

%\begin{lemma}\label{lem:fcalc}
%Let $X$ and $Y$ be independent random variables with the Binomial$(\frac{1}{1+p},r)$ distribution (where $1/1+p$ is the success probability and $r$ is the number of trials). Let $Z=X-Y$ and let  $f_B\left(k;\frac{1}{p+1},r\right)=P(Z=k)$ for integers $k=-r, \dots, 0, \dots r$. Then
%$$f_B\left(k;\frac{1}{p+1},r\right)=\frac{1}{(1+p)^{2r}}\sum\limits_{j=0}^r{r\choose |k|+j}{r\choose j}p^{2r-2j-|k|}$$
%\end{lemma}
%\begin{proof}
%\begin{eqnarray*}
%P(Z=k)&=&
%\begin{cases}
%\sum\limits_{j=0}^r P(X=j+k)P(Y=j)&\text{ if }k\geq 0\\
% \sum\limits_{j=0}^r P(X=j)P(Y=j+|k|) &\text{ if }k<0
%\end{cases}\\
%&=&\sum\limits_{j=0}^r\Bigg\{ {r\choose j+|k|}\left(\frac{1}{1+p}\right)^{j+|k|}\left(\frac{p}{1+p}\right)^{r-j-|k|}\\
%&&\qquad\qquad\times{r\choose j}\left(\frac{1}{1+p}\right)^{j}\left(\frac{p}{1+p}\right)^{r-j}\Bigg\}\\
%&=&\frac{1}{(1+p)^{2r}}\sum\limits_{j=0}^r {r\choose j+|k|}{r\choose j}p^{2r-2j-|k|}
%\end{eqnarray*}
%\end{proof}

\begin{lemma}\label{lem:fourierh}
Let $X$ and $Y$ be independent random variables with the Binomial$(\frac{p}{1+p},r)$ distribution (where $p/(1+p)$ is the success probability and $r$ is the number of trials). Let $Z=X-Y$ and let  $f_B\left(k;\frac{p}{p+1},r\right)=P(Z=k)$ for integers $k=-r, \dots, 0, \dots r$.
Define the function $h$ as $h(k)=(-1)^k f_B\left(k;\frac{p}{p+1},r\right)$. The Fourier series transform $\widehat{h}$ of $h$ (defined as $\widehat{h}(t)=\sum_{\ell=-\infty}^\infty h(\ell)e^{i\ell t}$) is equal to
\begin{eqnarray*}
\widehat{h}(t)=\frac{1}{(1+p)^{2r}}(1-pe^{it})^r(1-pe^{-it})^r
\end{eqnarray*}
\end{lemma}
\begin{proof}
Define the random variable $Y^\prime =-Y$. Then $X+Y^\prime=Z$. Thus
\begin{eqnarray*}
\widehat{h}(t)&=&\sum\limits_{\ell=-\infty}^\infty h(\ell)e^{i\ell t}\\
&=&\sum\limits_{\ell=-\infty}^\infty (-1)^\ell f_B\left(\ell;\frac{p}{p+1},r\right)e^{i\ell t}\\
&=&\sum\limits_{\ell=-\infty}^\infty (-1)^\ell e^{i\ell t}P(Z=\ell)\\
&=&\sum\limits_{\ell=-\infty}^\infty e^{i\ell t}(-1)^\ell\sum\limits_{j=-\infty}^\infty P(X=\ell-j)P(Y^\prime=j)\\
&=&\sum\limits_{\ell=-\infty}^\infty e^{i\ell t}\sum\limits_{j=-\infty}^\infty (-1)^{\ell-j}P(X=\ell-j)(-1)^jP(Y^\prime=j)\\
&=&\sum\limits_{\ell=-\infty}^\infty \sum\limits_{j=-\infty}^\infty (-1)^{\ell-j}e^{i(\ell-j)t}P(X=\ell-j)(-1)^je^{ijt}P(Y^\prime=j)\\
&=&\sum\limits_{j=-\infty}^\infty (-1)^je^{ijt}P(Y^\prime=j)\sum\limits_{\ell=-\infty}^\infty(-1)^{\ell-j}e^{i(\ell-j)t}P(X=\ell-j)
\end{eqnarray*}
Now,
\begin{eqnarray*}
\lefteqn{\sum\limits_{\ell=-\infty}^\infty(-1)^{\ell-j}e^{i(\ell-j)t}P(X=\ell-j)}\\
&=&\sum\limits_{\ell=-\infty}^\infty(-1)^{\ell}e^{i\ell t}P(X=\ell)\\
&=&\sum\limits_{\ell=0}^r(-1)^{\ell}e^{i\ell t}P(X=\ell)\\
&&\text{(Since $X$ can only be $0,\dots,r$)}\\
&=&\sum\limits_{\ell=0}^r(-1)^{\ell}e^{i\ell t} {r \choose \ell}\left(\frac{p}{1+p}\right)^\ell\left(\frac{1}{1+p}\right)^{r-\ell}\\
&=&\frac{1}{(1+p)^r}\sum\limits_{\ell=0}^r(-1)^{\ell}e^{i\ell t} {r \choose \ell}p^{\ell}\\
&=&\frac{1}{(1+p)^r}\sum\limits_{\ell=0}^r{r\choose \ell}(-pe^{it})^{\ell} \\
&=&\frac{1}{(1+p)^r}(1-pe^{it})^r\text{ by the Binomial theorem}
\end{eqnarray*}
Thus continuing our previous calculation,
\begin{eqnarray*}
\widehat{h}(t)&=&\sum\limits_{j=-\infty}^\infty (-1)^je^{ijt}P(Y^\prime=j) \frac{1}{(1+p)^r}(1-pe^{it})^r\\
&=&\sum\limits_{j=-r}^0 (-1)^je^{ijt}P(Y^\prime=j) \frac{1}{(1+p)^r}(1-pe^{it})^r\\
&&\text{(since $Y^\prime$ can only be $-r,\dots, 0$)}\\
&=&\sum\limits_{j=-r}^0 (-1)^je^{ijt}P(Y=-j) \frac{1}{(1+p)^r}(1-pe^{it})^r\\
&&\text{(since $Y^\prime=-Y$ )}\\
&=&\sum\limits_{j=0}^r (-1)^je^{-ijt}P(Y=j) \frac{1}{(1+p)^r}(1-pe^{it})^r\\
\end{eqnarray*}
Now, similar to what we did before, we can derive that $\sum\limits_{j=0}^r (-1)^je^{-ijt}P(Y=j)= \frac{1}{(1+p)^r}(1-pe^{-it})^r$ and therefore
\begin{eqnarray*}
\widehat{h}(t)=\frac{1}{(1+p)^{2r}}(1-pe^{it})^r(1-pe^{-it})^r
\end{eqnarray*}
\end{proof}

\begin{theorem}\emph{(Restatement and proof of Theorem \ref{thm:dnbrowcone}).}
A  bounded row vector $\vec{x}=(\dots, \linebreak[0]x_{-2}, \linebreak[0]x_{-1}, \linebreak[0]x_0, x_1, x_2, \dots)$ belongs to $\rowcone(\set{\mech_{DNB(p,r)}})$ if for all integers $k$,
$$\forall k:~~~\sum\limits_{j=-r}^r (-1)^j f_B\left(j;\frac{p}{1+p},r\right) x_{k+j} \geq 0$$
where $p$ and $r$ are the parameters of the differenced negative binomial distribution and $f_B(\cdot;p/(1+p),r)$ is the probability mass function of the difference of two independent binomial (not negative binomial) distributions whose parameters are $p/(1+p)$ (success probability) and $r$ (number of trials).
\end{theorem}
\begin{proof}
For convenience, define the function $h$ as follows:
\begin{eqnarray*}
h(j)&=&(-1)^jf_B\left(j;\frac{p}{1+p},r\right)
\end{eqnarray*}
Let $g_{NB}(\cdot;p,r)$ be the probability distribution function for the difference of two independent NB$(p,r)$ random variables. Then the matrix representation $M_{DNB(p,r)}$ of the differenced negative binomial mechanism $\mech_{DNB(p,r)}$ is the matrix whose rows and columns are indexed by the integers and whose entries are defined as:
\begin{eqnarray*}
(i,j)\text{ entry of }M_{DNB(p,r)}&=& g_{NB}(i-j;p,r)
\end{eqnarray*}
By Theorem \ref{thm:invcnfinf} we need to show that $M_{DNB(p,r)}$ is the inverse of $\frac{(1+p)^{2r}}{(1-p)^{2r}}H$ where $H$ is the matrix whose rows and columns are indexed by the integers and whose entries are defined as:
\begin{eqnarray*}
(i,j)\text{ entry of }H&=&h(i-j)=(-1)^{i-j}f_B\left(i-j;\frac{p}{1+p},r\right)
\end{eqnarray*}
(to see how Theorem \ref{thm:invcnfinf} is applied, note that each entry of the product $\vec{x}H$ has the form $\sum\limits_{j=-r}^r (-1)^j f_B\left(j;\frac{p}{1+p},r\right) x_{k+j} $).

Now, to show that $M_{DNB(p,r)}$ and $\frac{(1+p)^{2r}}{(1-p)^{2r}}H$ are inverses of each other, we note that
\begin{eqnarray}
\lefteqn{(i,j)\text{ entry of }\left(M_{DNB(p,r)}H\right)}\nonumber\\
&=&\sum\limits_{\ell=-\infty}^\infty g_{NB}(i-\ell;p,r)h(\ell-j)\nonumber\\
&=&\sum\limits_{\ell^\prime=-\infty}^\infty g_{NB}(i-j-\ell^\prime;p,r)h(\ell^\prime)\nonumber\\
&=&\sum\limits_{\ell^\prime=-r}^r g_{NB}(i-j-\ell^\prime;p,r)h(\ell^\prime)\label{eqn:convnb}
\end{eqnarray}
The last step follows from the fact that $f_B(\ell^\prime;p,r)$ and $h(\ell^\prime)$ are nonzero only when $\ell^\prime$ is between $-r$ and $r$ since $f_B(\cdot;p,r)$ is the probability mass function of the difference of two binomial random variables (each of which is bounded between $0$ and $r$).

Now, Equation \ref{eqn:convnb} is the definition of the convolution \cite{rudin} of $g_{NB}(\cdot;p,r)$ and $h$ at the point $i-j$. That is,
\begin{eqnarray*}
(g_{NB}(\cdot;p,r)\star h)(k)=\sum\limits_{\ell^\prime=-r}^r g_{NB}(k-\ell^\prime;p,r)h(\ell^\prime)
\end{eqnarray*}
and thus to show that  $M_{DNB(p,r)}$ and $\frac{(1+p)^{2r}}{(1-p)^{2r}}H$ are inverses of each other, we just need to show that the convolution of $g_{NB}(\cdot;p,r)$ and $h$ at the point $0$ is equal to $\frac{(1-p)^{2r}}{(1+p)^{2r}}$ and that the convolution at all other integers is $0$. In other words, we want to show that for all integers $k$,
\begin{eqnarray}
(g_{NB}(\cdot;p,r)\star h)(k) = \frac{(1-p)^{2r}}{(1+p)^{2r}}~\delta(k)\label{eqn:showconvdelta}
\end{eqnarray}
where $\delta$ is the function that $\delta(0)=1$ and $\delta(k)=0$ for all other integers. Take the Fourier series transform of both sides while noting two facts: (1) the Fourier series transform of $\delta$ is $\widehat{\delta}(t)=\sum\limits_{\ell=-\infty}^\infty \delta(\ell)e^{i\ell t}\equiv 1$, and  (2) the Fourier transform of a convolution is the product of the Fourier transforms \cite{rudin}. Then the transformed version of Equation \ref{eqn:showconvdelta}  becomes
\begin{eqnarray}
\widehat{g_{NB}}(t)~\widehat{h}(t) = \frac{(1-p)^{2r}}{(1+p)^{2r}}\widehat{\delta}(t) \equiv \frac{(1-p)^{2r}}{(1+p)^{2r}}\label{eqn:fouriershow}
\end{eqnarray}
for all real $t$, where $\widehat{g_{NB}}$, $\widehat{h}$, $\widehat{\delta}$ are the Fourier series transforms of $g_{NB}(\cdot;p,r)$, $h$, and $\delta$, respectively. Once we prove that Equation \ref{eqn:fouriershow} is true, this implies Equation \ref{eqn:showconvdelta} is true (by the inverse Fourier transform) which then implies that $M_{DNB(p,r)}$ and $\frac{(1+p)^{2r}}{(1-p)^{2r}}H$ are inverses of each other and this would finish the proof (by Theorem \ref{thm:invcnfinf}).

Thus our goal is to prove Equation \ref{eqn:fouriershow}.
The Fourier series transform (i.e. characteristic function), as a function of $t$, of the NB$(p,r)$ distribution is known to be:
\begin{eqnarray*}
\left(\frac{1-p}{1-pe^{it}}\right)^r
\end{eqnarray*}
so $g_{NB}(\cdot;p,r)$, being the difference of two independent negative binomial random variables, has the Fourier series transform (as a function of $t$)
\begin{eqnarray*}
\widehat{g_{NB}}(t)&=&\left(\frac{1-p}{1-pe^{it}}\right)^r \left(\frac{1-p}{1-pe^{-it}}\right)^r
\end{eqnarray*}
By Lemma \ref{lem:fourierh},
\begin{eqnarray*}
\widehat{h}(t)=\frac{1}{(1+p)^{2r}}(1-pe^{it})^r(1-pe^{-it})^r
\end{eqnarray*}
Thus Equation \ref{eqn:fouriershow} is true and we are done.
\end{proof}